\documentclass[12pt, draftclsnofoot, onecolumn]{IEEEtran}
\ifCLASSINFOpdf
   \usepackage[pdftex]{graphicx}
   \graphicspath{{img/pdf/}{img/jpeg/}}
   \DeclareGraphicsExtensions{.pdf,.jpeg,.png}
\else
   \usepackage[dvips]{graphicx}
   \graphicspath{{img/eps/}}
   \DeclareGraphicsExtensions{.eps}
\fi
\usepackage{todonotes}

\usepackage[cmex10]{amsmath}
\usepackage{booktabs}
\usepackage{amsfonts}
\usepackage{amssymb}
\usepackage{mathrsfs}
\usepackage{subfig}
\usepackage{wasysym}
\usepackage{color}
\usepackage{epsfig} 
\usepackage{epstopdf}
\usepackage{float}
\usepackage{amsthm}
\usepackage{mathtools}
\usepackage{bbm}

\usepackage[noadjust]{cite}

\allowdisplaybreaks
\definecolor{light-gray}{gray}{0.8}
\def\nbx{{\mathbf{x}}}

\def\nb0{{\mathbf{0}}}
\def\nb1{{\mathbf{1}}}

\def\ncalB{{\mathcal{B}}}

\def\ncalK{{\mathcal{K}}}
\def\ncalL{{\mathcal{L}}}

\def\ncalN{{\mathcal{N}}}

\def\ncalS{{\mathcal{S}}}

\def\nbbE{{\mathbb{E}}}

\def\nbbP{{\mathbb{P}}}

\def\nbbR{{\mathbb{R}}}

\def\nrmc{{\rm c}}
\def\nrmd{{\rm d}}

\def\nrmm{{\rm m}}

\newtheorem{lemma}{Lemma}
\newtheorem{thm}{Theorem}

\newtheorem{ndef}{Definition}

\newtheorem{prop}{Proposition}
\newtheorem{cor}{Corollary}

\newtheorem{remark}{Remark}
\newtheorem{assumption}{Assumption}
	
\def\figref#1{Fig.\,\ref{#1}}%
\def\E{\mathbb{E}}

\def\P{\mathbb{P}}
\def\pc{\mathtt{P_c}}

\def\R{\mathbb{R}}

\def\T{\beta}							

\def\sir{\mathtt{SIR}}
\def\ase{\mathtt{ASE}}

\def\calL{\mathcal{L}}

\def\Bx{{\mathcal{B}}^x}
\def\Bxx{{\mathcal{B}}^{x_0}}
\def\jx{y}

\def\Nx{{\mathcal{N}}^x}
\def\Nxo{{\mathcal{N}}^{x_0}}

 \def\yj{y}
 \def\yjx{y}
 
 \def \tx {y_0}
 \def \htx {h_0}

\def\rx{z_{1}}
\def\ry{z_{2}}

\def \hyxx {h_{y_{x_0}}}
\def \hyx {h_{y_x}}

\usepackage[font={small}]{caption}

\usepackage{setspace}	
\begin{document}
\title{Modeling and Performance Analysis of Clustered Device-to-Device Networks}
\author{Mehrnaz Afshang,~\IEEEmembership{Student Member,~IEEE,},
{Harpreet S. Dhillon,~\IEEEmembership{Member,~IEEE}}
 and
{Peter Han Joo Chong,~\IEEEmembership{Member,~IEEE}}
\thanks{M.~Afshang is with Wireless@VT, Department of ECE, Virgina Tech, Blacksburg, VA, USA and INFINITUS@NTU, Department of EEE, Nanyang Technological University, Singapore.
Email: mehrnaz@vt.edu. H. S.~Dhillon  is with  Wireless@VT, Department of ECE, Virgina Tech, Blacksburg, VA, USA. Email: hdhillon@vt.edu. P. H. J.~Chong is with INFINITUS@NTU, Department of EEE, Nanyang Technological University, Singapore. Email: Ehjchong@ntu.edu.sg.}
\thanks{This paper will be presented in part at the IEEE Globecom, San Diego, CA, 2015~\cite{AfsDhiC2015}. \hfill Last updated: \today.}}
\maketitle
\vspace*{-10pt}
\begin{abstract}
Device-to-device (D2D) communication enables direct communication between proximate devices thereby improving the overall spectrum utilization and offloading traffic from cellular networks. This paper develops a new spatial model for D2D networks in which the device locations are modeled as a Poisson cluster process. Using this model, we study the performance of a typical D2D receiver in terms of coverage probability under two realistic content availability setups: (i) content of interest for a typical device is available at a device chosen uniformly at random from the same cluster, which we term \emph{uniform content availability}, and (ii) content of interest is available at the $k^{th}$ closest device from the typical device inside the same cluster, which we term \emph{$k$-closest content availability}. Using these coverage probability results, we also characterize the area spectral efficiency ($\ase$) of the whole network for the two setups. A key intermediate step in this analysis is the derivation of the distributions of distances from a typical device to both the intra- and inter-cluster devices. Our analysis reveals that an  optimum number of D2D transmitters must be simultaneously activated per cluster in order to maximize $\ase$. This can be interpreted as the classical tradeoff between more aggressive frequency reuse and higher interference power. The optimum number of simultaneously transmitting devices and the resulting $\ase$ increase as the content is made available closer to the receivers. Our analysis also quantifies the best and worst case performance of clustered D2D networks both in terms of coverage and $\ase$. 
 \end{abstract}
\IEEEpeerreviewmaketitle
\section{Introduction}
\IEEEPARstart{E}{nabling} direct communication between devices located in close proximity, termed device-to-device (D2D) communications, has several benefits compared to the conventional approach of communicating through a base station in a cellular network \cite{feng2014device,asadi2013survey,tehrani2014device,doppler2009device}. First, the spectral efficiency of the direct link is typically much higher due to a smaller link distance. Second, this circumvents the need to establish an end-to-end link through a base station, thereby offloading traffic from cellular networks. Third, while the D2D network can be visualized as an {\em ad hoc} network, it incurs  a much lower protocol overhead due to the assistance it gets from the existing cellular network. All these benefits  make it an attractive component for both the current 4G and the future 5G networks \cite{doppler2009device,lin2014overview,song2015wireless}. Clearly, {\em content centric} nature of D2D communication opens up several exciting possibilities that were not quite possible with traditional cellular architecture. This is primarily driven by the spatiotemporal correlation in the content demand \cite{cha2007tube,cheng2008statistics,richier2014modelling}. In particular, when a device downloads a {\em popular} file, it can deliver it locally to its proximate devices whenever they need them~\cite{zhang2014social,6787081,ji2013wireless}. We term each such set of proximate devices as a {\em cluster}. The performance of a typical D2D link within each cluster will mainly depend upon where the content of interest is available with respect to the typical receiver and the number of other D2D links active in the network. Comprehensive modeling and performance analysis of this clustered D2D network using tools from stochastic geometry is the main goal of this paper.%
  \subsection{Motivation and Related Work} 
Modeling and analysis of D2D communication has taken two main directions in the literature. The first one focuses on characterizing the scaling of per-device throughput as a function of the network size; see \cite{6787081,ji2013wireless,ji2014fundamental} for a small subset. To maintain analytical tractability, {\em protocol model} is typically assumed under which the transmission between two devices is successful only if (i) the distance between them is smaller than a certain predefined value, and (ii) there are no other active transmitters in the immediate neighborhood of the receiver. While these assumptions are somewhat restrictive, more general results are typically obtained by means of extensive simulations. The second direction, which is also more relevant to our work, focuses on characterizing metrics, such as the distribution of per-device throughput and coverage probability, using tools from stochastic geometry under more general physical layer models in which the metrics are defined in terms of the actual received powers from the desired and interfering devices, as opposed to Euclidean distances that appear in the protocol model discussed above. While stochastic geometry has successfully been applied to study various aspects of {\em ad hoc} and cellular networks over the past decade; see~\cite{haenggi2012stochastic,baccelli2009stochastic,dhillon2012modeling,mukherjee2012distribution,novlan2013analytical,elsawy2014stochastic} for a small subset, it has also been used more recently to study D2D networks~\cite{lin2013comprehensive,elsawy2014analytical,andreev2014analyzing,feng2014tractable,sun2014d2d,george2014analytical,sakr2014cognitive,mungara2014spatial,
lin2013modeling}. We discuss these works in more detail next.
 

Mode selection in D2D-enabled uplink cellular networks where traffic can be offloaded from the cellular network to the D2D network is studied in \cite{lin2013comprehensive,elsawy2014analytical,andreev2014analyzing}. Interference management between the cellular and D2D networks is investigated in~\cite{feng2014tractable,sun2014d2d,george2014analytical,sakr2014cognitive,mungara2014spatial}. The performance of the multicast transmission wherein each D2D transmitter (D2D-Tx) has a common massage for all the intended D2D receivers (D2D-Rxs) inside the cluster is analyzed in~\cite{lin2013modeling}.  
The common approach in all these works is to model the locations of the the D2D-Txs as a Poisson Point Process (PPP) while two approaches are considered for modeling the locations of the D2D-Rxs. In the first approach, to lend analytical tractability, the network is modeled using a Poisson Dipole Process (PDP) where the D2D-Rxs are located at a fixed distance from the D2D-Txs~\cite{feng2014tractable,sun2014d2d,george2014analytical,sakr2014cognitive,mungara2014spatial,lin2013comprehensive}. Although this is a good first order model, the assumption of fixed link distance is quite restrictive. This assumption is relaxed 
 by assuming that the intended D2D-Rx is uniformly distributed within a circle around its serving D2D-Tx \cite{elsawy2014analytical ,andreev2014analyzing,lin2013modeling}. However, neither of these stochastic geometry-based approaches captures the possibility of having multiple proximate devices any of which can act as a serving device  for a given device, which is quite fundamental to D2D networks \cite{6787081,zhang2014social,ji2013wireless}. In this paper, we address these shortcomings by developing a new and more realistic spatial model for D2D networks in which the devices form clusters. We consider {\em out-band} D2D in which D2D and cellular transmissions do not interfere with each other. 
 More details of the model along with the other main contributions of this paper are provided next.
 %

%

\subsection{Contributions and Outcomes}
%

{\em Realistic tractable model for D2D networks.} We develop a new and more realistic way of modeling D2D networks in which the device locations are modeled as a Poisson cluster process, in particular a variant of a Thomas cluster process~\cite{haenggi2012stochastic}. This is unlike the popular approaches where the device locations are assumed to be uniform over the plane, such as in the PPP and PDP models discussed above. The proposed model captures the fact that a given device typically has multiple proximate devices any of which can potentially act as a serving device. Using tools from stochastic geometry, we characterize the performance of this D2D network for two content availability scenarios: (i) {\em uniform content availability}, where the content of interest for a typical device is available at a device chosen uniformly at random from the same cluster, and (ii) {\em $k$-closest content availability}, where the content of interest is available at the $k^{th}$ closest device to a typical device inside the same cluster. It should be noted that while Poisson cluster process has been used in the literature to model wireless networks, in particular see \cite{ganti2009interference}, the performance analysis is usually performed at a point that may not necessarily be a part of the point process. Unlike this scenario, we perform analysis at a {\em typical device}, which by definition is a part of the cluster process. This setup brings forth new technical challenges, e.g., the need to characterize the distribution of distances from the typical device to other devices, as discussed below.
 
   \begin{figure}[t!] 
\centering{
        \includegraphics[width=.50\textwidth]{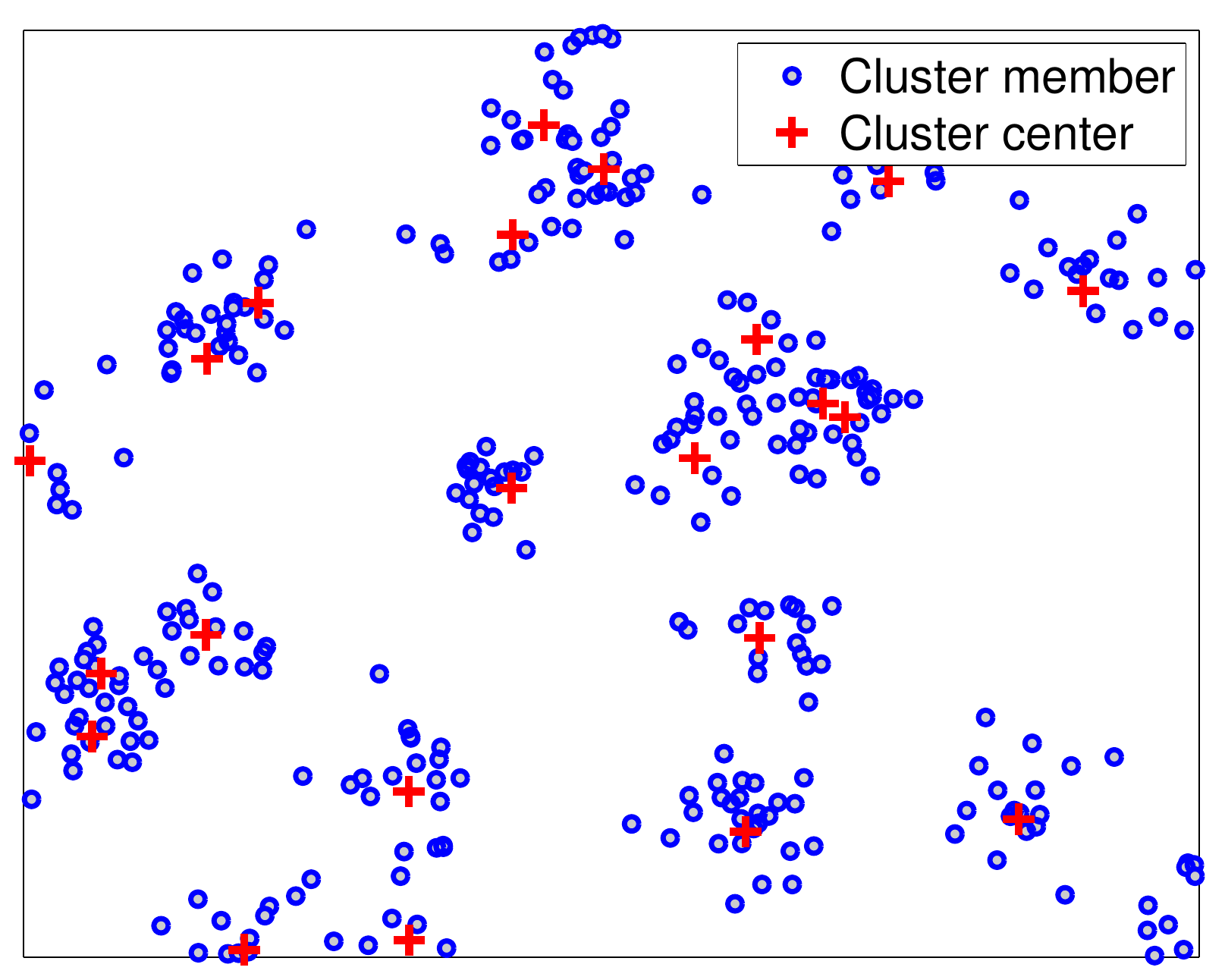}
              \caption{Proposed D2D cluster model where devices are normally distributed around each cluster center.}
                \label{Fig: Network Topo}
                }
\end{figure}

{\em New distance distributions enabling $\ase$ and coverage probability analysis.} 
We derive exact expressions for coverage probability of a typical device and $\ase$ of the whole network for the two scenarios described above. Several closed-form bounds and approximations are also derived. As intermediate results, we characterize the distributions of the distances of the typical device to its serving device and intra- and inter-cluster interfering devices in the two cases. The key enabler of our analysis is the observation that conditioned on the appropriate random variables, the distances from a typical device to any other randomly chosen device in a Thomas cluster process are independent and Rician distributed. The $k$-closest content availability case additionally requires the use of order statistics in order to characterize the serving distance distribution and a more careful treatment of the intra-cluster interference field. Using these distributions, we derive Laplace transforms of both intra- and inter-cluster interference powers, using which the coverage probability and $\ase$ are characterized. %

{\em System design insights.} Our analysis leads to several system design guidelines. First, it reveals the existence of the optimal number of links that must be activated per cluster in order to maximize $\ase$. This can be interpreted as the classical tradeoff between more aggressive frequency reuse and higher interference power. The optimum number of simultaneously active links per cluster and the resulting $\ase$ increase as the content is made available closer to the receivers. For typical operational regimes of interest for D2D networks, our results reveal that significant gains can be achieved by activating optimum number of links compared to strictly orthogonal strategy in which only one link per cluster is active. The $k^{th}$ closest content availability strategy also allows to characterize the best and worst case performances of a clustered D2D network in terms of coverage probability and $\ase$ by tuning the value of $k$, where the best and worst cases correspond to the when the content is available at the closest and farthest device, respectively. 

 
\section{System Model} \label{Sec: System Model}

We assume that each device has a certain content that can be requested by the other devices in the {\em same} cluster and the devices across clusters need not communicate. This can be justified in a practical setup from multiple perspectives, two of which are: (i) the inter-cluster distances will be typically much larger than the intra-cluster distances between devices, thus making it easier to communicate within a given cluster, and (ii) the devices in one cluster may not have information that is of interest to the devices in the other clusters. For example, devices forming a cluster in a sports bar are more likely to be interested in sports-related content as opposed to devices in an academic setting, such as a library, where the nature of {\em popular} content may be entirely different. While the tools developed in this paper can be extended to handle the case where communication across clusters is allowed, it is not in the scope of this paper and is left as a promising future work. We now discuss the key modeling details for this setup.

\subsection{Spatial Setup and Key Assumptions}
The locations of the devices are modeled by a {\em Poisson cluster process}, where the parent point process is modeled by a PPP $\Phi_{\nrmc}$ with density $\lambda_{\nrmc}$, and the offspring point processes (one per parent) are conditionally independent~\cite{DalVerB2003}. The union of all the offspring points constitutes a Poisson cluster process. The parent and offspring points will be henceforth referred to as the {\em cluster centers} and the {\em cluster members} (or simply {\em devices}), respectively. The cluster members are assumed to be independent and identically distributed (i.i.d.) according to a symmetric normal distribution with variance $\sigma^2$  around each cluster center $x \in \Phi_{\nrmc}$. Therefore, the density function of the device location $y \in \nbbR^2$ relative to a cluster center is
\begin{equation}\label{Eq: notmaldist}
   f_Y(\yj)= \frac{1}{2 \pi \sigma^2}\exp \left(-\frac{\|\yj\|^2}{2 \sigma^2}\right).
\end{equation}
If the number of devices in each offspring process, i.e., cluster, were Poisson distributed, this process is simply a {\em Thomas cluster process}~\cite{haenggi2012stochastic}. However, to simplify certain {\em order statistics} arguments in the sequel, we assume that the total number of devices per cluster is fixed and equal to $N$. As will be evident from the discussion below, the number of simultaneously active devices will still be different across clusters, thereby providing sufficient generality to the model.%



The proposed model is illustrated in \figref{Fig: Network Topo}. The set of all devices in a cluster $x \in \Phi_{\nrmc}$, denoted by $\Nx$, is partitioned randomly into two subsets: (i) set of possible {\em transmitting devices} denoted by $\Nx_{\rm t}$, and (ii) set of possible {\em receiving devices} denoted by $\Nx_{\rm r}$. The set of simultaneously transmitting devices in this cluster is denoted by $\Bx \subseteq \Nx_{\rm t}$, where $|\Bx|$ is assumed to be Poisson distributed with mean $\bar{m}$ conditioned on $|\Bx| \leq |\Nx_{\rm t}|$. Note that in the limiting case, about half of the devices in each cluster will transmit to the other half. Therefore, for notational convenience, we assume that the total number of transmitting devices per cluster is limited to $M=N/2$. This can be easily relaxed in case one wants to allow all $N$ devices to be transmitting in a certain application. Without loss of generality, we perform analysis for a {\em typical device}, which is a randomly chosen device in a randomly chosen cluster, termed {\em representative cluster}, inside the network. Assuming the cluster center of the representative cluster to be located at $x_0 \in \Phi_c$, the typical device by definition is in $\Nxo_{\rm r}$. Due to the stationarity of this process, we assume that the typical device is located at the origin. Since the performance of the D2D link to this receiver depends upon where in the cluster is the data available, we consider following two setups to model content availability in the representative cluster:
\begin{enumerate}
\item \emph{Uniform content availability}. The content of interest for the typical device is available at a device that is chosen uniformly at random from the set $\Nxo_{\rm t}$ in the same cluster.
\item \emph{$k$-closest content availability.} The content of interest to the typical device is located at its  $k^{th}$ closest device from the set $\Nxo_{\rm t}$ in the same cluster. By tuning the value of $k$, the content can be biased to lie closer (small $k$) or farther (large $k$) from the typical device.
\end{enumerate}
After fixing the location of the serving device as per one of the strategies above, the intra-cluster interfering devices are sampled uniformly at random from the remaining $M-1$ devices in $\Nxo_{\rm t}$ in the representative cluster. Since a representative cluster has a serving device by definition, for concreteness we assume that the number of interfering devices is Poisson distributed with mean $\bar{m}-1$, which means that a total of $\bar{m}$ devices are active on average in this cluster. Similarly, the inter-cluster interfering devices are sampled uniformly at random from the set of transmitting devices of each cluster, such that the number of active devices in each cluster is Poisson distributed with mean $\bar{m}$ conditioned on the total being less than $M$.%

\subsection{Channel Model}
Recall that the cluster center of the representative cluster is assumed to be located at $x_0 \in \Phi_{\nrmc}$. Assume that D2D-Tx of interest ({\em serving} transmitter) is located at $\tx$ inside that cluster. The distance of this D2D-Tx from the typical device (D2D-Rx of interest) is denoted by $r=\|\tx+x_0\|$, where $r$ is a realization of random variable $R$ whose distribution depends upon the content availability strategy as discussed in detail in the sequel. Assuming transmit power of each device to be $P_{\rm d}$, the received power at D2D-Rx of interest is
\begin{equation}\label{Eq: typical received power}
    P= P_\nrmd \htx \|x_0+\tx\|^{-\alpha},
\end{equation}
where $\htx\sim \exp(1)$  is i.i.d. exponential random variable which models Rayleigh fading and $\alpha$ is path loss exponent. Incorporating shadowing is left as a promising direction of future work. To define interference field, it is useful to define the set of all simultaneously active D2D-Txs as: 
\begin{equation}\label{Eq: complete process }
    \Psi_\nrmm= \cup_{x \in \Phi_\nrmc} \Bx,
\end{equation}
where recall that $\Bx$ is the set of simultaneously active D2D-Txs inside a cluster $x \in \Phi_{\nrmc}$. In this network, the total interference  caused at the D2D-Rx of interest can be written as the sum of two independent terms: (i)  intra-cluster interference caused by the interfering D2D-Txs inside the representative cluster, and (ii)  inter-cluster interference caused by simultaneously active D2D-Txs outside the representative cluster. Recalling that the D2D-Tx of interest is located at $\tx$  with respect to the cluster center $x_0$, the intra-cluster interference power can be expressed as: %
\begin{align}\label{Eq: Intra_cluster_Int_typical}
        I_\mathrm{Tx-cluster}=\sum_{\jx\in \Bxx \setminus \tx} P_\nrmd \hyxx\|x_0+\yjx\|^{-\alpha}.
\end{align}
Similarly, the  interference  from the simultaneously active D2D-Txs outside the representative cluster, $x_0$, at the  D2D-Rx of interest can be expressed as:
\begin{align}\label{Eq: Intera_cluster_Int_typical}
        I_\mathrm{\Psi_\nrmm\setminus Tx-cluster}=\sum_{x\in \Phi_{\nrmc}\setminus x_0}\sum_{\jx \in \Bx} P_{\nrmd}\hyx\|x+\yj\|^{-\alpha}.
\end{align}
Denoting the total interference power experienced by the typical device by $ I_{\Psi_\nrmm}$ and recalling the serving distance to be $r=\|\tx+x_0\|$, the $\sir$ experienced by the typical device is
\begin{align}\label{Eq: SINR}
\sir(r)= \frac{P_\nrmd \htx r^{-\alpha}}{I_{\Psi_\nrmm}} = \frac{P_\nrmd \htx r^{-\alpha}}{I_{\Psi_\nrmm \setminus \mathrm{Tx-cluster}}+I_\mathrm{Tx-cluster}}.
\end{align}
For notational simplicity, we assume that the system operates in the interference limited regime, i.e., the background noise is negligible compared to the  interference and is hence ignored. This means that the transmit power term cancels in the $\sir$ expression above and can hence be ignored, i.e., we can set $P_{\nrmd} = 1$ without any loss of generality. For a quick reference, the notation used in this paper is summarized in Table~\ref{table:notation}.



\section{Distribution of the Distances}
This is the first main technical section of the paper where we characterize the distributions of the distances from the typical device to various intra- and inter-cluster devices. These distance distributions will be used in the analysis of coverage probability and $\ase$ in the next section. We first focus on the {\em unordered case} in which we characterize the distribution of the distances from a typical device to a device chosen uniformly at random in a given cluster $x\in \Phi_\nrmc$. Using this, we characterize the inter-cluster distance distributions for both the content availability cases, and serving and intra-cluster distance distributions for the {\em uniform content availability} case. We then analyze the {\em ordered case} in which the distances from a typical device to the devices of a given cluster $x\in \Phi_\nrmc$ are ordered in the increasing order. Using this, we characterize the serving and intra-cluster distances to the typical device in the {\em $k$-closest content availability} case.

Before going into more technical details, we define the functional forms of the probability density functions (PDFs) of the Rayleigh and Rician distributed random variables, which will significantly simplify the notation in the rest of this section. 
\begin{ndef}[Rayleigh distribution]
The PDF  of the Rayleigh distributed random variable is
\begin{equation}
\mathtt{Raypdf}(a;\sigma^2)= f_A(a)= \frac{a}{ \sigma ^2}\exp\left(-\frac{a^2}{2 \sigma^2}\right), \quad a>0,
\end{equation}
 where $\sigma$ is the scale parameter of the distribution.
\end{ndef}
\begin{ndef}[Rician distribution]
The PDF of the Rician distributed random variable is
\begin{equation}
\mathtt{Ricepdf}(a,b;\sigma^2)= f_A(a|b)=\frac{a}{ \sigma^2} \exp\left(-\frac{a^2+b^2}{2 \sigma^2}\right) I_0\left(\frac{a b}{\sigma^2}\right), \quad a>0,
\end{equation}
where $I_0(.)$ is  the modified Bessel function with order zero and  $\sigma$ is the scale parameter.
\end{ndef}

\begin{table}
\centering{
\caption{Summary of notation}
\scalebox{.8}{%
\begin{tabular}{c|c}
  \hline
   \hline
  \textbf{Notation} & \textbf{Description}  \\
     \hline
  $\Phi_\nrmc; \lambda_\nrmc$ & Independent PPP modeling the locations of D2D cluster centers; density of $\Phi_\nrmc$\\
  \hline
 $\Nx$& Set of devices inside the cluster centered at $x \in \Phi_\nrmc$\\
    \hline
  $N=|\Nx|$& Total number of devices per cluster (assumed same for each cluster)\\
  \hline   
 $\Nx_{\rm t}; \Nx_{\rm r}$& Subsets of $\Nx$ denoting the set of possible transmitting and receiving devices\\
    \hline
  $\Bx \subseteq \Nx_{\rm t}$; $\bar{m}$ &Set of simultaneously active devices inside the cluster with mean $\bar{m}$\\
  \hline
      $\ncalS^{x}_{\rm t}$ &Set of distances from the typical device to the inter-cluster devices\\
    \hline
  $\Bxx \subseteq \Nxo_{\rm t}$; $\bar{m}$ &Set of simultaneously active devices inside the cluster with mean $\bar{m}$ in the cluster $x_0 \in \Phi_\nrmc$\\
  \hline
      $\ncalS^{x_0}_{\rm t}$ &Set of distances from the typical device to the intra-cluster devices\\
  \hline
    $\Psi_\nrmm= \cup_{x \in \phi_\nrmc} \Bx$ &Set of all simultaneously active devices in the D2D cluster network\\
  \hline
   $M = |\Nx_{\rm t}|$& Number of maximum possible transmitting devices per cluster, $M=N/2$\\
  \hline
 $\sigma^2$  & Scattering variance of the cluster member locations around each cluster center\\
 \hline
  $P_\nrmd$  & Transmit power of each device engaged in D2D communications\\
 \hline
 $\alpha$  & Path loss exponent for all wireless links; $\alpha>2$\\
 \hline
 $h_{y_x}$  & Channel power gain under Rayleigh fading where $h_{y_x} \sim \exp (1)$ \\
 \hline
 $\pc;\T;\mathtt{ASE}$   & Coverage probability; target $\sir$; area spectral efficiency\\
 \hline
  \hline
\end{tabular} \label{table:notation}
}
} 
\end{table}

\subsection{Distance Distributions in the Unordered Case}
Lets start our discussion with the intra-cluster distances by focusing on the representative cluster. Denote by $\ncalS^{x_0}_{\rm t}$, the set $\{S_i\}_{i=1:M}$ of distances from the typical device to the set of possible transmitting devices $\Nxo_{\rm t}$ in the cluster $x_0 \in \Phi_\nrmc$, where $s_i=\|x_0+y\|$  is the realization of $S_i$. The ordering in this case is arbitrary, which means $S_i \in \ncalS^{x_0}_{\rm t}$ will be interpreted as the distance from the typical device to a device chosen uniformly at random from $\Nxo_{\rm t}$. Whenever this interpretation is clear, we will drop index $i$ from $s_i$ and $S_i$. Characterizing the marginal distribution of $S$ is quite straightforward. Since $x_0$ and $y$ are i.i.d.  Gaussian random variables with variance $\sigma^2$, $x_0 + y$ is also Gaussian with variance $2\sigma^2$. Therefore, $S$ is Rayleigh distributed with probability density function (PDF) 
$f_{S}(s)= \mathtt{Raypdf}(s;2\sigma^2)$.  
However, this does not completely characterize $\ncalS^{x_0}_{\rm t}$ because it doesn't capture  the fact that the distances from intra-cluster devices to the typical device $\{\|x_0+y\|\}$ are correlated due to the common factor $x_0$. That being said, if we condition on the location of the cluster center $x_0$ relative to the typical device, the distances in the set $\ncalS^{x_0}_{\rm t}$ are i.i.d. since the device locations are i.i.d. around the cluster center by assumption. This conditional distribution is characterized in the following Lemma. In the proof, we show that instead of conditioning on the location $x_0$, a ``weaker'' conditioning on the distance $\nu_0 = \|x_0\|$, suffices. Therefore, the statement of the Lemma is presented in terms of $\nu_0$.

 \begin{lemma}[Distribution of i.i.d. sequence $\ncalS^{x_0}_{\rm t}$] \label{lem:Intra_cluster location _typical_general}
Conditioned on the distance $\nu_0 = \|x_0\|$, the PDF of an element $S$ chosen uniformly at random from the i.i.d. sequence $\ncalS^{x_0}_{\rm t}$ is
\begin{align}\label{Eq: Intra_cluster location _typical}
f_{S}(s| x_0) = f_{S}(s| \nu_0)= \mathtt{Ricepdf}(s,\nu_0;\sigma).
\end{align}
\end{lemma}
\begin{proof}
See Appendix \ref{App: proof of Ricain}.  
\end{proof}

\begin{remark}[Serving and intra-cluster distances for uniform content availability]
As discussed above, the $M$ elements of $\ncalS^{x_0}_{\rm t}$ are i.i.d. with the distribution characterized by Lemma~\ref{lem:Intra_cluster location _typical_general}. For the uniform content availability case, one of the elements of $\ncalS^{x_0}_{\rm t}$ is chosen uniformly at random as the serving distance, and the rest correspond to the distances to the possible intra-cluster interfering devices. Since the elements of $\ncalS^{x_0}_{\rm t}$ were ``unordered'' and the selection of serving device was done uniformly at random, all these distances are i.i.d. and follow Rician distribution given by \eqref{Eq: Intra_cluster location _typical}. The results are stated formally as Corollaries of Lemma~\ref{lem:Intra_cluster location _typical_general} below. 
\end{remark}

\begin{cor}[Uniform content availability: serving distance] \label{lem:Intra_cluster location _typical}
For the uniform content availability case, the conditional PDF of the serving distance $r=\|x_0+\tx\|$, conditioned on the distance $\nu_0 = \|x_0\|$ between the cluster center and the typical device, is $f_R(r|\nu_0)= \mathtt{Ricepdf}(r,\nu_0;\sigma)$.  
  \end{cor}
  
 \begin{cor}[Uniform content availability: intra-cluster interferer distance]
  \label{lem:intra_cluster_uniform_interference}
For the uniform content availability case, the distances from the intra-cluster interfering devices to the typical device $\{w=\|x_0+\yjx\|, \forall \jx\in \Bxx \setminus \tx\}$ are conditionally i.i.d., conditioned on $\nu_0 = \|x_0\|$, with each distance following the PDF given by $f_{W}(w| \nu_0)= \mathtt{Ricepdf}(w,\nu_0;\sigma)$.
 \end{cor}

We now look at the distribution of the distances from inter-cluster devices to the typical device. Recall that in both the content availability strategies, the inter-cluster interfering devices are chosen uniformly at random from the set of transmitting devices $\Nx_{\rm t}$ in  each cluster $x \in \Phi_{\nrmc}$. Therefore this discussion is applicable to both the strategies. Denoting the distances from inter-cluster interfering devices of the cluster $x \in \Phi_{\nrmc}$ to the typical device by $\ncalS^{x}_{\rm t}$, it can be shown that the elements of $\ncalS^{x}_{\rm t}$ are conditionally i.i.d., conditioned on the distance $\nu = \|x\|$ form the typical device to the cluster center $x \in \Phi_{\nrmc}$. It follows on the same lines as Lemma~\ref{lem:Intra_cluster location _typical_general}, except that conditioning here is on $\nu = \|x\|$ and not $\nu_0 = \|x_0\|$. The result is formally stated below. 

\begin{lemma} [Inter-cluster interferer distance distribution] Conditioned on the distance $\nu = \|x\|$ between the cluster center $x \in \Phi_{\nrmc}$ and the typical device, the distances from the inter-cluster interfering devices to the typical device $\{u=\|x+\yj\|, \forall \yj \in \Bx\}$ are i.i.d. with each element following the PDF given by $f_{U}(u| \nu)= \mathtt{Ricepdf}(u,\nu;\sigma)$.
\label{Lem: Inter cluster distance}
\end{lemma}
\begin{proof}
Recall that every $\yj \in \Bx$ is an independent zero-mean Gaussian random variable. Hence conditional on the common distance $\nu = \|x\|$ between the cluster center $x \in \Phi_{\nrmc}$ and the typical device, the distances $\{u=\|x+\yj\|, \forall \yj \in \Bx\} \equiv \ncalS^{x}_{\rm t}$ are i.i.d. and the conditional PDF of each element of $\ncalS^{x}_{\rm t}$ can be derived on the same lines as Lemma \ref{lem:Intra_cluster location _typical_general}.
\end{proof}


\subsection{Distance Distributions in the Ordered Case} \label{subsection:OrderedDistance}
While the unordered case was sufficient to handle the uniform content availability strategy and the inter-cluster interfering distances in both the strategies, we need to consider the ``ordered case'' to handle the serving and intra-cluster interfering distances in the $k$-closest content availability strategy. For this analysis, consider the set of distances $\{S_i\}_{i=1:M}$, denoted by $\ncalS^{x_0}_{\rm t}$, from the previous subsection, that denotes the ``unordered'' distances from the typical device to the set of possible transmitting devices $\Nxo_{\rm t}$ in the cluster $x_0 \in \Phi_\nrmc$ with sampling distribution $f_S(s|\nu_0)$ given by Lemma~\ref{lem:Intra_cluster location _typical_general}. Different from the previous subsection, we order the elements of $\ncalS^{x_0}_{\rm t}$ in terms of the increasing distance from the typical device and denote them by $\{S_{(i)}\}_{i=1:M}$, where $S_{(1)}\leq ... \leq S_{(k)} ... \leq S_{(M)}$. 
In this case the serving distance $R$ corresponds to $S_{(k)}$ whose distribution can be derived using order statistics. The result is presented in the following Lemma.

%
\begin{lemma}[$k$-closest content availability: serving distance] \label{lem:k-closest_serving}
For the $k$-closest content availability strategy, the conditional distribution of serving distance, conditioned on $\nu_0 = \|x_0\|$, is
\begin{equation} \label{eq: k-closet dis to typical}
f_R(r| \nu_0)=\frac{M!}{(k-1)!(M-k)!}{F_S(r|\nu_0)}^{k-1}(1-F_S(r|\nu_0))^{M-k}  f_S(r| \nu_0),
\end{equation}
\normalsize
with $ f_S(r| \nu_0)=\mathtt{Ricepdf}(r,\nu_0;\sigma)$ derived in Lemma~\ref{lem:Intra_cluster location _typical_general} being the conditional PDF of $S$, and $F_S(r|\nu_0)=1-Q_1(\frac{\nu_0}{\sigma}, \frac{r}{\sigma})$ being the conditional cumulative distribution function (CDF), where $Q_1(a,b)$ is the Marcum Q-function defined as $Q_1(a,b) = \int_b^{\infty} t e^{-\frac{t^2 + a^2}{2}} I_0(at) {\rm d} t$.
\end{lemma}
\begin{proof} In Lemma~\ref{lem:Intra_cluster location _typical_general}, we showed that the unordered sequence of distances $\{S_i\}_{i=1:M}$ is conditionally i.i.d. with conditional PDF of each element given by $f_S(r| \nu_0)$. The result simply follows from the PDF of the $k^{th}$ order statistic of this sequence of i.i.d. random variables~\cite[eqn.~(3)]{david1970order}. 
\end{proof}

We next characterize the distances from the interfering devices to the typical device in the $k$-closest content availability strategy. In comparison to the uniform content availability strategy, this is more involved since the $k^{th}$ closest device is fixed {\em a priori} as the serving device and hence cannot act as an interferer. To address this issue, we divide  the set of simultaneously active devices into three subsets, $\Bxx\equiv \{{\mathcal{B}}^{x_0}_{\rm in}, y_0,{\mathcal{B}}^{x_0}_{\rm out}\}$, where  the serving device is located at a distance  $s_{(k)}=\|x_0+y_0\|$ from the typical device, and   ${\mathcal{B}}^{x_0}_{\rm in}$  (${\mathcal{B}}^{x_0}_{\rm out}$)  denote the set of devices that are closer (farther) to the typical device compared to the serving device. We show that conditional on the distance $\nu_0 = \|x_0\|$, the distances from the typical device to the devices in ${\mathcal{B}}^{x_0}_{\rm in}$ are i.i.d. and their distribution is characterized in the Lemma below. The same holds for the devices in ${\mathcal{B}}^{x_0}_{\rm out}$. This i.i.d. property will play a major role in the exact analysis of Laplace transform of intra-cluster interference in $k$-closest content availability case in the next section.

 \begin{lemma}[$k$-closest content availability: intra-cluster interferer distance]\label{OK_Last_Lemma}
For the $k$-closest content availability strategy,


a) the distances from the devices in the set ${\mathcal{B}}^{x_0}_{\rm in}$ to the typical device, i.e., $\{w_{\rm in}=\|x_0+\yjx\|\}$, are conditionally i.i.d., conditioned on the serving distance $r$ and the distance $\nu_0 = \|x_0\|$ between the cluster center and the typical device, with each distance following the PDF
\begin{equation} \label{eq: f_w_in}
 f_{W_{\rm in}}(w_{\rm in}|\nu_0, r)=\left\{
 \begin{array}{cc}
 \frac{f_{S}(w_{\rm in}|\nu_0)}{F_{S}(r|\nu_0)}, & w_{\rm in}<r\\
 0, & w_{\rm in}\geq r
 \end{array}\right.,
 \end{equation}
where $ f_S(w_{\rm in}| \nu_0)=\mathtt{Ricepdf}(w_{\rm in},\nu_0;\sigma)$, and $F_S(r|\nu_0)=1-Q_1(\frac{\nu_0}{\sigma},\frac{r}{\sigma})$, and
 
b) the distances from the devices in the set ${\mathcal{B}}^{x_0}_{\rm out}$ to the typical device, i.e., $\{w_{\rm out}=\|x_0+\yjx\|\}$, are conditionally i.i.d., conditioned on the serving distance $r$ and the distance $\nu_0 = \|x_0\|$ between the cluster center and the typical device, with each distance following the PDF
 \begin{equation}\label{eq: f_w_out}
 f_{W_{\rm out}}(w_{\rm out}|\nu_0, r)=\left\{
 \begin{array}{cc}
 \frac{f_{S}(w_{ \tt out}|\nu_0)}{1-F_{S}(r|\nu_0)}, & w_{\rm out}>r \\
0 & w_{\rm out}\leq r
\end{array}\right.,
 \end{equation}
where $ f_S(w_{\rm out}| \nu_0)=\mathtt{Ricepdf}(w_{\rm out},\nu_0;\sigma)$, and $F_S(r|\nu_0)=1-Q_1(\frac{\nu_0}{\sigma},\frac{r}{\sigma})$.
 \end{lemma}
\begin{proof}
See Appendix \ref{App: proof of inter intra-cluster distance}.
\end{proof}


\section{Coverage Probability and $\ase$ Performance}
Using the distance distributions derived in the previous section, we now derive the coverage probability of the typical device and the $\ase$ of the whole network for the two content availability strategies. We begin our discussion with the uniform content availability strategy. 

\subsection{ Uniform Content Availability}
In this subsection, we focus on the case in which the content of interest for the typical device is available at a device chosen uniformly at random in the representative cluster. As evident in the sequel, we need the Laplace transforms of the intra- and inter-cluster interference powers as the intermediate results for the coverage and $\ase$ analysis. These are derived next.

\subsubsection{Laplace Transform of Interference}
We start by deriving  exact expressions and several bounds and approximations on the Laplace transform of  intra-cluster interference. 
\begin{lemma}[Laplace transform of intra-cluster interference] \label{lem: lap intra typical}
Under uniform content availability, the conditional Laplace transform of the intra-cluster interference power given by~\eqref{Eq: Intra_cluster_Int_typical}, conditioned on the distance $\nu_0=\|x_0\|$ between the cluster center and the typical device is 
\begin{align}\label{Eq: laplace intra typical _sum}
 \calL_{I_{ \mathrm{Tx-cluster}}} (s |\nu_0)=\sum_{k=0}^{M-1} \Big ( \int_0^\infty \frac{1}{1+s w^{-\alpha}}f_{W}(w|\nu_0)\nrmd w\Big)^k
 \frac{(\bar{m}-1)^k e^{-(\bar{m}-1)}}{k! \xi},
\end{align}
with   $\xi=\sum_{j=0}^{M-1}\frac{(\bar{m}-1)^j e^{-(\bar{m}-1)}}{j!}$.  For $M \gg \bar{m}$, the above expression reduces to
\begin{align}\label{Eq: laplace intra typical}
\calL_{I_{ \mathrm{Tx-cluster}}} (s |\nu_0) = \exp\left(-(\bar{m}-1)\int_0^\infty \frac{s w^{-\alpha}}{1+s w^{-\alpha}}f_W(w|\nu_0)\nrmd w\right),
\end{align}
where $f_W(w| \nu_0)$ is given by Corollary \ref{lem:intra_cluster_uniform_interference}.
\end{lemma}
\begin{proof} See Appendix \ref{App: Proof of Lemma  lap intra typical}.
\end{proof}
\begin{remark}
Note that the assumption $M\gg \bar{m}$, under which a simplified expression is derived in the above Lemma, is applicable when the number of simultaneously active devices per cluster is much smaller compared to the cluster size. As discussed in the sequel, this is also the regime in which the network performance in terms of $\ase$ will be usually optimized, especially for the uniform content availability case. Therefore, the simpler expression will be treated as a proxy of the exact expression for the derivation of simpler bounds and approximations.
\label{Rem: asymptotic case for the number of devices}
\end{remark}

While Lemma~\ref{lem: lap intra typical} provides an exact expression for the Laplace transform of intra-cluster interference, it is usually desirable to derive simple but tight approximations and bounds whenever possible to draw useful system design insights, which we do next. First we derive an approximation for the Laplace transform of intra-cluster interference under the following assumption. 
\begin{assumption}[Uncorrelated intra-cluster distances]\label{Ass: Identical intra-cluster distances}
Recall that the distances between intra-cluster devices and typical device, denoted by $\ncalS^{x_0}_{\rm t}$, are correlated due to the common factor $x_0$. However, the coverage analysis can be simplified significantly if this correlation is ignored, which we do as a part of this assumption. More formally, we assume that the serving and intra-cluster interferer distances are i.i.d. Rayleigh distributed with marginal distributions $f_R(r)=\mathtt{Raypdf}(r;2\sigma^2)$ and $f_W(w)=\mathtt{Raypdf}(w; 2\sigma^2)$, respectively. In other words, we do not condition on the distance $\nu_0 = \|x_0\|$, which simplifies the coverage analysis by allowing separate deconditioning on the serving and intra-cluster interferer distances as discussed in the sequel.%
\end{assumption}
Under the above assumption, the Laplace transform of intra-cluster interference at the typical device is stated as a the following Corollary of Lemma~\ref{lem: lap intra typical}. 
\begin{cor}[Approximation]
 \label{cor: Intra_cluster distance _typical approximation}
Under Assumption~\ref{Ass: Identical intra-cluster distances}, the Laplace transform of intra-cluster interference at the typical device is
\begin{align}\label{Eq:Lap intra typical app}
 \tilde{\calL}_{I_{ \mathrm{Tx-cluster}}} (s ) = \exp\left(-(\bar{m}-1)\int_0^\infty \frac{s w^{-\alpha}}{1+s w^{-\alpha}}f_{W}(w)\nrmd w\right),
\end{align}
\normalsize
where $f_W(w)=\mathtt{Raypdf}(w; 2\sigma^2)$.
\end{cor}
 By applying Jensen's inequality on the result of Lemma~\ref{lem: lap intra typical},  we also provide a closed form lower bound on the Laplace transform of intra-cluster interference in the next Corollary.
\begin{cor}[Lower bound]
\label{Cor: Lap intra typical low bound}
 The lower bound on Laplace transform of intra-cluster interference at the typical device is
\begin{align}\label{Eq: Lap intra typical low bound}
    \calL_{I_{ \mathrm{Tx-cluster}}} (s )\ge \exp\left(-  \frac{(\bar{m}-1)}{4  \sigma^2} s^{\frac{2}{\alpha}}  \frac{2 \pi / \alpha}{\sin (2 \pi / \alpha )} \right).
\end{align}
\end{cor}
\begin{proof}
See Appendix \ref{App: proof of intra typical low bound}.
\end{proof}
 We will use this result along with the lower bound on the Laplace transform of inter-cluster interference  (will be derived later in this section) to derive a  closed form approximation for coverage probability and $\ase$. The results  will provide insights into several system design guidelines. We now state the exact result for the Laplace transform of inter-cluster interference.
\begin{lemma}[Laplace transform of inter-cluster interference]\label{Lem: Lap_Inter}
The Laplace transform of inter-cluster interference experienced at the typical device, given by~\eqref{Eq: Intera_cluster_Int_typical}, is
\begin{align}\label{Eq: Lap_Inter_sum}
\calL_{I_{\Psi_\nrmm \setminus \mathrm{Tx-cluster}}} (s)=\exp\Big(-2 \pi \lambda_\nrmc\int_0^\infty \Big(1-\sum_{k=0}^{M} \Big ( \int_0^\infty \frac{1}{1+s u^{-\alpha}}f_{U}(u|\nu)\nrmd u\Big)^k  
 \frac{\bar{m}^k e^{-\bar{m}}}{k! \eta} \Big)\nu \nrmd \nu\Big),
\end{align}
with   $\eta=\sum_{j=0}^{M}\frac{\bar{m}^je^{-\bar{m}}}{j!}$.  For $M \gg \bar{m}$, the above expression reduces to
\begin{align}\label{Eq: Lap_Inter}
\calL_{I_{\Psi_\nrmm \setminus \mathrm{Tx-cluster}}} (s) &= \exp\Big(-2 \pi \lambda_\nrmc\int_0^\infty \Big(1-\exp\Big(-\bar{m}\int_0^\infty \frac{s u^{-\alpha}}{1+s u^{-\alpha}}
  f_U(u|\nu)\nrmd u \Big)\Big)\nu \nrmd \nu\Big),
\end{align}
where $f_{U}(u| \nu)$  given by Lemma \ref{Lem: Inter cluster distance}. 
\end{lemma}
\begin{proof} See Appendix~\ref{APP:Proof of Lemma 2}.
\end{proof}
\begin{remark}Recall that the inter-cluster interfering devices are chosen uniformly at random from each cluster in both the content availability strategies, which means that the above result is applicable for both the uniform and $k$-closest content availability strategies.
\end{remark}
Under the assumption $M \gg \bar{m}$, a closed form lower bound for the Laplace transform of inter-cluster interference can be derived, which is stated next. 
\begin{cor}[Lower bound]   \label{cor: worst case inter lower}
Using \eqref{Eq: Lap_Inter} from Lemma~\ref{Lem: Lap_Inter}, the following lower bound on the Laplace transform of inter-cluster interference experienced by a typical device can be derived: 
 \begin{equation}\label{Eq: Lp_Inter_dPPP_Lbound}
         \calL_{I_{\Psi_\nrmm \setminus \mathrm{Tx-cluster}}} (s)\ge \exp\left(- \pi \lambda_\nrmc \bar{m} s^{2/\alpha}  \frac{2 \pi / \alpha}{\sin (2 \pi / \alpha )}\right).
 \end{equation}
 \end{cor}
\begin{proof}
See Appendix~\ref{App: Proof of Cor1}.
\end{proof} 
\subsubsection{Coverage Probability}
The coverage probability is formally defined as the probability that $\sir$ experienced by the typical device exceeds a certain pre-determined threshold $\T$ for successful demodulation and decoding at the receiver. It is mathematically expressed as:
\begin{align}
   \pc = \mathbb{E}_R\left[\,\mathbb{P}\{\mathtt{SIR}(R)>\T\,|\,R\}\right]. 
\end{align}
Using the Laplace transform expressions of intra- and inter-cluster interference powers derived so far in this section, an exact expression for $\pc$ is derived in the following Theorem.
\begin{thm}[Coverage probability]\label{thm: coverage typical uniform}
Using Laplace transform of intra- and inter-cluster interference, given respectively by Lemmas~\ref{lem: lap intra typical} and~\ref{Lem: Lap_Inter}, 
the coverage probability is
\begin{align}\label{Eq: coverage typical}
  \pc=  \int_0^\infty  \int_0^\infty &\ncalL_{I_{\psi_\nrmm \setminus \mathrm{Tx-cluster}}}(\T r^{\alpha})\ncalL_{I_\mathrm{Tx-cluster}}(\T r^{\alpha}|\nu_0) 
  f_R(r| \nu_0) f_{V_0}(\nu_0) \nrmd r \nrmd \nu_0,
\end{align}
\normalsize
where $f_{V_0}(\nu_0)=\mathtt{Raypdf}(\nu_0;\sigma^2)$.
\end{thm}
\begin{proof}  
From the definition of coverage probability, we have
\begin{align}
&\pc = \mathbb{E}_R\left[ \nbbP \left\{ \frac{h_{0x_0} r^{-\alpha}}{I_{\Psi_\nrmm \setminus \mathrm{Tx-cluster}}+I_\mathrm{Tx-cluster}} > \T \,\Big|\,R  \right\} \right]\nonumber \\ 
&\stackrel{(a)}{=} \nbbE_R \left[ \nbbE\left[\exp\left(- \T r^\alpha (I_{\Psi_\nrmm \setminus \mathrm{Tx-cluster}}+I_\mathrm{Tx-cluster}) \right) \Big| R \right]\right], \nonumber
\end{align}
where $(a)$ follows from Rayleigh fading assumption. The result now follows from the independence of intra-cluster and inter-cluster interference, followed by de-conditioning over $R$ given $\nu_0$ using the serving link distance distribution given by Corollary \ref{lem:Intra_cluster location _typical}, followed by de-conditioning over $\nu_0$, which is simply a Rayleigh distributed random variable due to the position being sampled from a Gaussian distribution in $\nbbR^2$ around each cluster center.
\end{proof}
\subsubsection{Area Spectral Efficiency}
The $\ase$ simply denotes the average number of bits transmitted  per unit time per unit bandwidth per unit area. Assuming that all the D2D-Txs use Gaussian codebooks for their transmissions, we can use Shannon's capacity formula to define $\mathtt{ASE}= \lambda \log_2(1+\T) \pc$, where $\lambda$ is the density of the active transmitters and $\pc$ is the coverage probability of the typical device. The result specialized to our setup is given in the following theorem. 
\begin{prop} The $\ase$ of the clustered D2D network under uniform content availability is
\begin{equation}\label{eq: ASE}
    \mathtt{ASE}=\bar{m}\lambda_\nrmc \log_2(1+\T) \pc,
\end{equation}
where $\pc$ is given by~\eqref{Eq: coverage typical} and $\bar{m}\lambda_\nrmc$ represents the average density of simultaneously active D2D-Txs inside the network.
\end{prop}
\begin{remark}[Optimum number of simultaneously active links]
Note that there is a clear  tradeoff between link efficiency and cluster interference. While more active links means potentially higher $\mathtt{ASE}$, it also increases interference significantly. $\ase$ can, in principle, be maximized as
\begin{equation}\label{Eq: ASE optimization}
  \mathtt{ASE}^*= \max _{\bar{m}\in {1,...,M}} \bar{m}\lambda_\nrmc \log_2(1+\T) \pc. 
\end{equation}
By  solving this $\ase$ optimization problem numerically, we will demonstrate the existence of an optimal value of $\bar{m}$ that maximizes the $\ase$ in the numerical results section. 
\end{remark}
\subsubsection{Bounds and approximations}
After characterizing coverage and $\ase$ exactly, we now focus on tight bounds and approximations that will result in easy-to-compute expressions providing useful system design guidelines. First, we derive coverage probability under Assumption~\ref{Ass: Identical intra-cluster distances}. 
\begin{cor} [Coverage probability approximation] \label{cor: coverage typical approximation} 
Under Assumption~\ref{Ass: Identical intra-cluster distances}, the coverage probability of the typical device in the uniform content availability case is 
\begin{align}\label{Eq: Coverage typical app}
    \pc = \int_0^\infty\tilde{\calL}_{I_{ \mathrm{Tx-cluster}}} (\T r^{\alpha})
    \ncalL_{I_{\psi_\nrmm \setminus \mathrm{Tx-cluster}}}(\T r^{\alpha})  f_{R}(r) \nrmd r, 
\end{align}
where $f_{R}(r)=\mathtt{Raypdf}(r;2\sigma^2)$, $\tilde{\calL}_{I_{ \mathrm{Tx-cluster}}} (\T r^{\alpha})$ is the Laplace transform of intra-cluster interference under Assumption~\ref{Ass: Identical intra-cluster distances} given by \eqref{Eq:Lap intra typical app}, and $\ncalL_{I_{\psi_\nrmm \setminus \mathrm{Tx-cluster}}}(\T r^{\alpha})$ is the Laplace transform of inter-cluster interference given by \eqref{Eq: Lap_Inter}. 
\end{cor}
\begin{proof}
The proof follows on the same lines as Theorem \ref{thm: coverage typical uniform}, except the fact that due to Assumption~\ref{Ass: Identical intra-cluster distances}, there is no conditioning on $\nu_0$ (dependence in the intra-cluster distances is ignored), which allows to decondition on the serving and intra-cluster interfering distances separately. 
\end{proof}
In the numerical result section, we show that this easy-to-compute approximation is also very tight. To gain more insights into the behavior of this clustered network, we derive closed form approximations for both coverage and $\ase$ in the following corollaries. 
\begin{cor}[Coverage: closed-from approximation]\label{Cor: closed form lower bound for typical user} The closed form approximation for the coverage probability of a typical device in uniform content availability case is
\begin{align}\label{Eq: closed from Pc lower}
    \pc\simeq \frac{1}{(4 \pi \lambda_{\nrmc}\sigma^2 \bar{m}+\bar{m}-1) \beta^{\frac{2}{\alpha}}\frac{2 \pi / \alpha}{\sin (2 \pi / \alpha )}+1}.
\end{align}
\end{cor}
\begin{proof}  The proof  follows from the expectation of the lower bounds provided in~\eqref{Eq: Lap intra typical low bound} and~\eqref{Eq: Lp_Inter_dPPP_Lbound} when $s= \T r^\alpha$, with respect to the  marginal distribution of serving link distance  $f_{R}(r)=\mathtt{Raypdf}(r;2\sigma^2)$. This is an approximation because of independent de-conditioning over serving link and intra-cluster interferer distance distributions.
\end{proof}
\begin{cor}[$\ase$: closed-form approximation] \label{cor:ASE_closedform}
The $\ase$ in this case can be approximated as
\begin{align}\label{Eq: closed form ASE lower}
    \ase \simeq \frac{\bar{m}\lambda_{\nrmc}}{(4 \pi  \lambda_{\nrmc}\sigma^2 \bar{m}+\bar{m}-1) \beta^{\frac{2}{\alpha}}\frac{2 \pi / \alpha}{\sin (2 \pi / \alpha )}+1}.
\end{align}
\end{cor}
\begin{proof}
The proof follows from the definition of the $\ase$  when $\pc$ is substituted by \eqref{Eq: closed from Pc lower}.
\end{proof}
In the numerical results sections, we show that these closed-form expressions provide surprisingly tight approximations for both coverage and $\ase$. Based on these two approximations (Corollaries \ref{Cor: closed form lower bound for typical user} and \ref{cor:ASE_closedform}), we make the following two observations: (i) $\ase$ and $\pc$ are decreasing functions of $\sigma^2$, which means that the coverage probability and $\ase$ will increase if the devices form more {\em dense} clusters around cluster centers, as opposed to more spread-out clusters, and (ii) increasing density of  cluster centers  has a conflicting effect on the coverage probability and $\ase$: coverage decreases and $\ase$ increases, which means that the density of cluster centers can be increased to enhance $\ase$ as long as the coverage probability remains acceptable. 
\subsection{$k$-Closest Content Availability}
 \label{sec:generalized}
In this setup, we  extend the analysis to the case where content of interest is available  at $k^{th}$ closest transmitting device to the typical device. Recall that by tuning the value of $k$, the content can be biased to lie closer (small $k$) or farther (large $k$) from the typical device. The best and worst case performances can be characterized by setting $k=1$ and $k=M$, respectively. Similar to the previous subsection, we begin by deriving the Laplace transform of interference. 
%
%
\subsubsection{Laplace transform of interference}
Using the distance distributions derived in Section~\ref{subsection:OrderedDistance}, the exact expression for the Laplace transform of intra-cluster interference is derived in the Lemma below. As discussed in the previous section, this case is significantly more involved than the uniform content availability case due to the way the serving device is selected.
\begin{lemma}[Laplace transform of intra-cluster interference]\label{lem: Laplace intra k-closest} Under $k$-closest content availability, the conditional Laplace transform of the intra-cluster interference power given by~\eqref{Eq: Intra_cluster_Int_typical}, conditioned on $\nu_0=\|x_0\|$, is $\ncalL_{ I_{\mathrm{Tx-cluster}}}(s,r|\nu_0)=$
\begin{align}\label{eq: Lap-closest content}
\sum_{n=0}^{M-1} \sum_{l=0}^{g_{\rm m}} {n \choose l}  p^{l} {(1-p)}^{n-l} \frac{1}{I_{1-p}(n-g_{\rm m},g_{\rm m}+1)} \ncalK_{\rm in}(s,r|\nu_0 ) ^l  
\ncalK_{\rm out}(s, r|\nu_0)^{n-l} p_n
\end{align}
\begin{align*}
&\text{with,} &\ncalK_{\rm in}(s, r|\nu_0)&=\int_0^r \frac{1}{1+s w_{\rm in}^{-\alpha} }f_{W_{\rm in}}(w_{\rm in}|\nu_0, r) \nrmd w_{\rm in}\\
&&\ncalK_{\rm out}(s, r|\nu_0)&=\int_r^{\infty} \frac{1}{1+s w_{\rm out}^{-\alpha} } f_{W_{\rm out}}(w_{\rm out}|\nu_0, r) \nrmd w_{\rm out},
\end{align*}
where $g_{\rm m}=\min(n,k-1)$, $p=\frac{k-1}{M-1}$, $p_n=\frac{(\bar{m}-1)^n e^{-(\bar{m}-1)}}{n! \xi}$ with $\xi=\sum_{j=0}^{M-1}\frac{(\bar{m}-1)^j e^{-(\bar{m}-1)}}{j!}$,   $I_{1-p}$ is regularized incomplete beta function,  $f_{W_{\rm in}}(w_{\rm in}|\nu_0, r)$ and $ f_{W_{\rm out}}(w_{\rm out}|\nu_0, r)$  are given by \eqref{eq: f_w_in} and  \eqref{eq: f_w_out} respectively. Note that here  zero to the zero power is defined as one.
\end{lemma}

\begin{proof}
See Appendix \ref{App: Laplace intra k-closest}.
\end{proof}
\begin{cor}[Best link: $1^{st}$ closest] The Laplace transform of intra-cluster interference for best link, i.e., $k=1$,  under condition of $\bar{m}\ll M$  reduces to
\begin{align}
\ncalL_{ I_{\mathrm{Tx-cluster}}}(s,r|\nu_0)=\exp\left(-(\bar{m}-1)\int_r^{\infty} \frac{s w_{\rm out}^{-\alpha}}{1+s w_{\rm out}^{-\alpha} } f_{W_{\rm out}}(w_{\rm out}|\nu_0, r) \nrmd w_{\rm out}\right),
\end{align}
where $f_{W_{\rm out}}(w_{\rm out}|\nu_0,r)$  is given by \eqref{eq: f_w_out}.
\end{cor}
\begin{cor}[Worst link: $M^{th}$ closest] The Laplace transform of intra-cluster interference for worst link, i.e., $k=M$, under condition of $\bar{m}\ll M$  reduces to
\begin{align}
\ncalL_{ I_{\mathrm{Tx-cluster}}}(s,r|\nu_0)=\exp\left(-(\bar{m}-1)\int_0^{r} \frac{s w_{\rm in}^{-\alpha}}{1+s w_{\rm in}^{-\alpha} }f_{W_{\rm in}}(w_{\rm in}|\nu_0, r) \nrmd w_{\rm in}\right),
\end{align}
where $f_{W_{\rm in}}(w_{\rm in}|\nu_0, r)$  is given by \eqref{eq: f_w_in} .
\end{cor}
\subsubsection{Coverage Probability Analysis}
Using the serving distance distribution derived in Lemma~\ref{lem:k-closest_serving} and the inter- and intra-cluster interference Laplace transforms derived respectively by Lemmas~\ref{Lem: Lap_Inter} and~\ref{lem: Laplace intra k-closest}, the exact expression for the coverage probability in this case is derived below.
\begin{thm}[$k$-closest content availability: coverage probability]\label{thm: coverage typical}
The coverage probability of the typical device for the $k$-closest content availability strategy is
\begin{align}\label{Eq: coverage typical_k_closet}
  \pc=  \int_0^\infty  \int_0^\infty &\ncalL_{I_{\psi_\nrmm \setminus \mathrm{Tx-cluster}}}(\T r^{\alpha})\ncalL_{ I_{\mathrm{Tx-cluster}}}(s,r|\nu_0) 
  f_R(r| \nu_0) f_{V_0}(\nu_0) \nrmd r \nrmd \nu_0,
\end{align}
\normalsize
where $f_R(r|\nu_0)$ is the serving distance distribution given by \eqref{eq: k-closet dis to typical}, $f_{V_0}(\nu_0)=\mathtt{Raypdf}(\nu_0;\sigma^2)$, and the inter- and intra-cluster interference Laplace transforms are given by Lemmas~\ref{Lem: Lap_Inter} and~\ref{lem: Laplace intra k-closest}.
\end{thm}
\begin{proof}
The proof follows on the same line as Theorem \ref{thm: coverage typical uniform} and is hence skipped.
\end{proof}


\subsubsection{Area Spectral Efficiency}
Unlike uniform content availability, where each active receiver selects a serving device uniformly at random from its cluster, defining $\ase$ for the $k$-closest content availability strategy is more involved because different devices may have a different value of $k$ for their serving devices depending upon the scheduling strategy. Since we are not characterizing scheduling policies explicitly, the information about $k$ for each device is not known. Therefore, we derive $\ase$ for a simpler case in which all active receivers connect to the $k^{th}$ closest device, with $k$ being the same for all the devices. Two special cases of interest are (i) $k=1$: all devices connect to the closest device from the set of possible transmitting devices in their clusters, and (ii) $k=M$: all devices connect to the farthest devices from the set of possible transmitting devices in their clusters. The $\ase$ for this setup is given by the following Theorem.
\begin{prop}[$k$-closest content availability: $\ase$]
The $\ase$ for this setup is
\begin{equation}\label{eq: ASE general}
    \mathtt{ASE}=\bar{m}\lambda_\nrmc \log_2(1+\T) \pc,
\end{equation}
where $\pc$ is given by~\eqref{Eq: coverage typical_k_closet}  and $\bar{m}\lambda_\nrmc$ represents the average density of simultaneously active D2D-Txs. The maximum and minimum $\ase$ is achieved for $k=1$ and $k=M$, respectively.
\end{prop}
Due to the dependence of $\pc$, given by~\eqref{Eq: coverage typical_k_closet}, on $k$, $\ase$ is also parameterized by $k$. As noted above, this allows us to characterize the maximum and minimum $\ase$ achievable in a network.%

\subsubsection{Bounds and approximations}
From Lemma~\ref{lem: Laplace intra k-closest}, we note that the exact expression of intra-cluster interference for general $k$ involves two summations, which further complicates the numerical evaluation of the exact coverage probability expression given in Theorem~\ref{thm: coverage typical}. Therefore, we propose two simple but tight approximations. In the first approximation, we ignore the effect of excluding the $k^{th}$ closest device (serving) from the field of possible interferers. This allows us to assume that the intra-cluster interferers are chosen uniformly at random, resulting in the same Laplace transform of intra-cluster interference that we had in the uniform content availability case given in Lemma~\ref{lem: lap intra typical}. The resulting approximation from coverage probability is given below.

\begin{cor}[$k$-closest content availability: coverage approximation]\label{thm: coverage typical general}Using Laplace transform of interference in~\eqref{Eq: laplace intra typical} and~\eqref{Eq: Lap_Inter}, the coverage probability can be approximated as 
\begin{align}\label{Eq: coverage typical general}
  \pc \simeq  \int_0^\infty  \int_0^\infty &\ncalL_{I_{\psi_\nrmm \setminus \mathrm{Tx-cluster}}}(\T r^{\alpha})\ncalL_{I_\mathrm{Tx-cluster}}(\T r^{\alpha}|\nu_0) 
  f_R(r| \nu_0) f_{V_0}(\nu_0) \nrmd r \nrmd \nu_0,
\end{align}
\normalsize
where $f_R(r| \nu_0)$  given by \eqref{eq: k-closet dis to typical} and $f_{V_0}(\nu_0)=\mathtt{Raypdf}(\nu_0;\sigma^2)$.
\end{cor}
\begin{proof}
The proof follows on the same lines as the proof of Theorem \ref{thm: coverage typical uniform} and hence is skipped.
\end{proof}
In addition to the assumption made while deriving Corollary~\ref{thm: coverage typical general}, if we also ignore the correlation among distances between intra-cluster devices to the typical device (see Assumption \ref{Ass: Identical intra-cluster distances}), we can use the simpler result for the Laplace transform of intra-cluster interference given by Corollary~\ref{cor: Intra_cluster distance _typical approximation}. The resulting coverage probability approximation is stated next.
\begin{cor} [$k$-closest content availability: coverage approximation] \label{cor: coverage typical approximation k closest}
 Using the results of approximation given by \eqref{Eq:Lap intra typical app} and Laplace transform of inter-cluster interference in \eqref{Eq: Lap_Inter},  the coverage probability  can be approximated as
  \begin{align}\label{Eq: Coverage typical app k-closest}
    \pc\simeq  \int_0^{\infty} \tilde{\calL}_{I_{ \mathrm{Tx-cluster}}} (\T r^{\alpha})
    \ncalL_{I_{\psi_\nrmm \setminus \mathrm{Tx-cluster}}}(\T r^{\alpha})  f_{R}(r) \nrmd r
\end{align}
  $$\text{with, }f_{R}(r)=\frac{M!}{(k-1)!(M-k)!}{F_S(r)}^{k-1}(1-F_S(r))^{M-k} f_S(r),$$
  \normalsize
  where $f_S(r)=\mathtt{Raypdf}(r;2\sigma^2) $, and $F_S(r)=1-\exp(-\frac{r^2}{4 \sigma^2})$.
\end{cor}
In the numerical results section, we will show that both these approximations are remarkably tight and can in fact be treated as proxies of the exact result if needed.

\section{Results and Discussion}
\subsection{Uniform Content Availability}
\subsubsection{Validation of results}
\label{Sub: Validation of results}
\begin{figure}[t!]
\centering{
        \includegraphics[width=.5\textwidth]{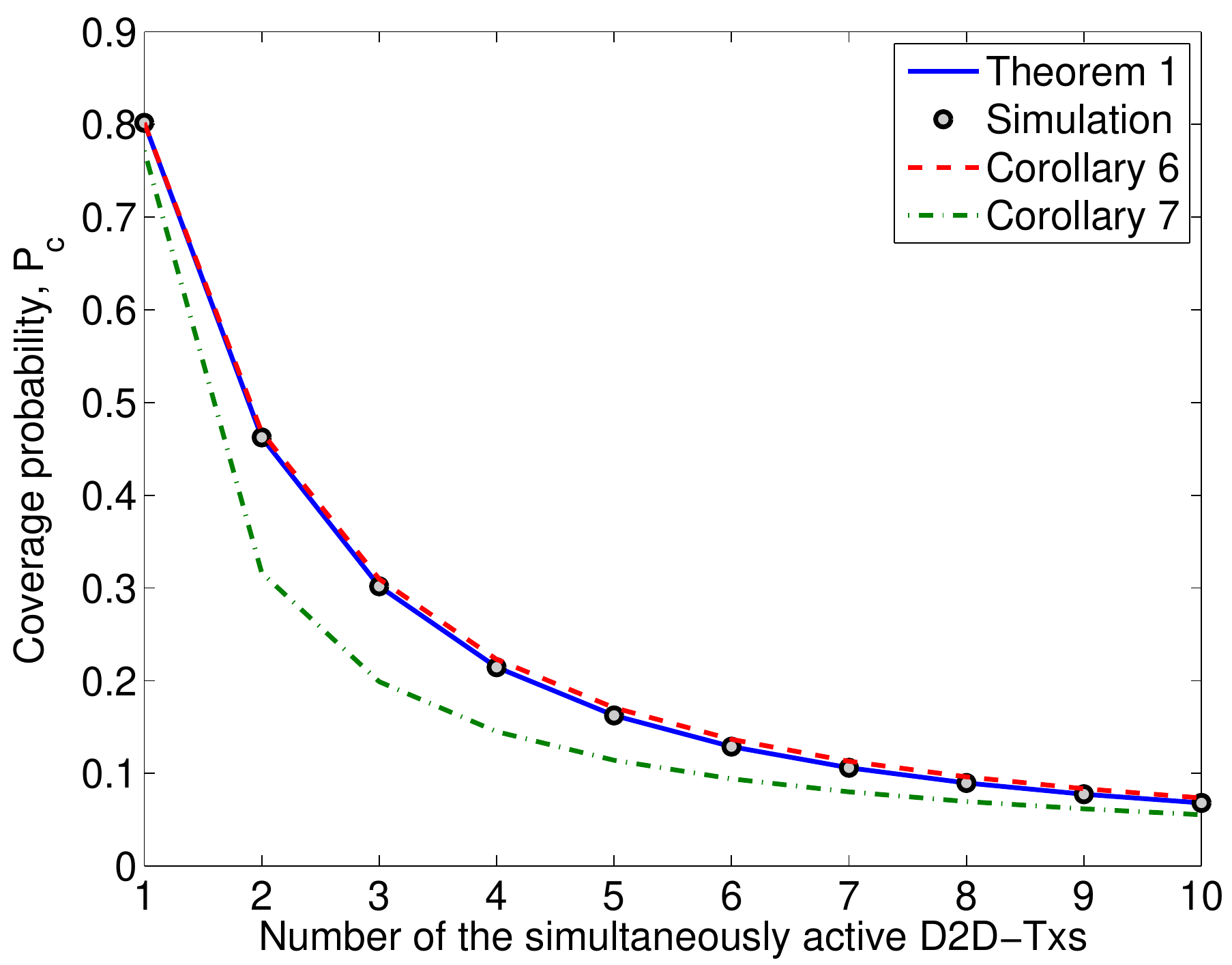}
               \caption{Coverage probability as a function of number of simultaneously active D2D-Txs in the uniform content availability case when  $\sigma=10$ and $\lambda_\nrmc=150$ clusters/km$^2$.}
                \label{Fig: validtion of result intra coverage of typical and central user}
                }
\end{figure}
We first validate the analytical results and investigate the tightness of various approximations derived for coverage probability in the uniform content availability case. In all simulations, the locations of devices are drawn from a Poisson cluster process over a $50  \times 50$ $\mathrm{km^2}$ square region. The cluster centers are spatially distributed as a PPP, and the devices are normally scattered around them. We set the $\sir$ threshold for the successful demodulation and decoding, $\beta$, as 0 dB. Comparing the analytical and the simulation results in \figref{Fig: validtion of result intra coverage of typical and central user}, we note that they are a perfect match, thereby validating our analysis. The results also show that the approximation for coverage probability given by Corollary \ref{cor: coverage typical approximation} is surprisingly tight. The closed-form approximation given by Corollary \ref{Cor: closed form lower bound for typical user} also provides a decent lower-bound.
 \subsubsection{Impact of inter- and intra-cluster interference on the
coverage probability}
We plot the coverage probability in the presence of (i) only inter-cluster interference, (ii) only intra-cluster interference, and (iii) total interference for different scattering variances in \figref{Fig: Coverage probability versus number of simultaneously active D2D-Tx for  lambdac=150} and for different densities of cluster centers in \figref{Fig: Coverage probability v.s. number of active D2D-Tx for different lambda sigma50}. It is easy to infer that the coverage probability of the typical device is strongly dictated by the intra-cluster interference.
 Interestingly, the coverage probability in the absence of inter-cluster interference is independent of the scattering variance in \figref{Fig: Coverage probability versus number of simultaneously active D2D-Tx for  lambdac=150}.
  This is because scattering variance has two counter-balancing effects that cancel each other exactly: decreasing scattering variance increases intra-cluster interference by reducing the link distances to the interfering devices, and improves the desired link quality by again reducing the distance to the serving device. For  coverage probability in the absence of the intra-cluster interference, it can be seen that the smaller scattering variance provides higher coverage probability, which is also true for the coverage probability computed in the presence of both inter- and intra-cluster interference.  
In comparison to the setup of \figref{Fig: Coverage probability versus number of simultaneously active D2D-Tx for  lambdac=150}, we decrease the cluster center density and increase the scattering variance in \figref{Fig: Coverage probability v.s. number of active D2D-Tx for different lambda sigma50}. While smaller cluster-center density reduces inter-cluster interference and hence improves coverage, increasing cluster variance reduces coverage (as discussed above in the context of Fig.~\ref{Fig: Coverage probability versus number of simultaneously active D2D-Tx for  lambdac=150}), which has a dominant effect in Fig.~\ref{Fig: Coverage probability v.s. number of active D2D-Tx for different lambda sigma50}.

\begin{figure}
\noindent \begin{minipage}{.49\textwidth}
  \includegraphics[width=1\textwidth]{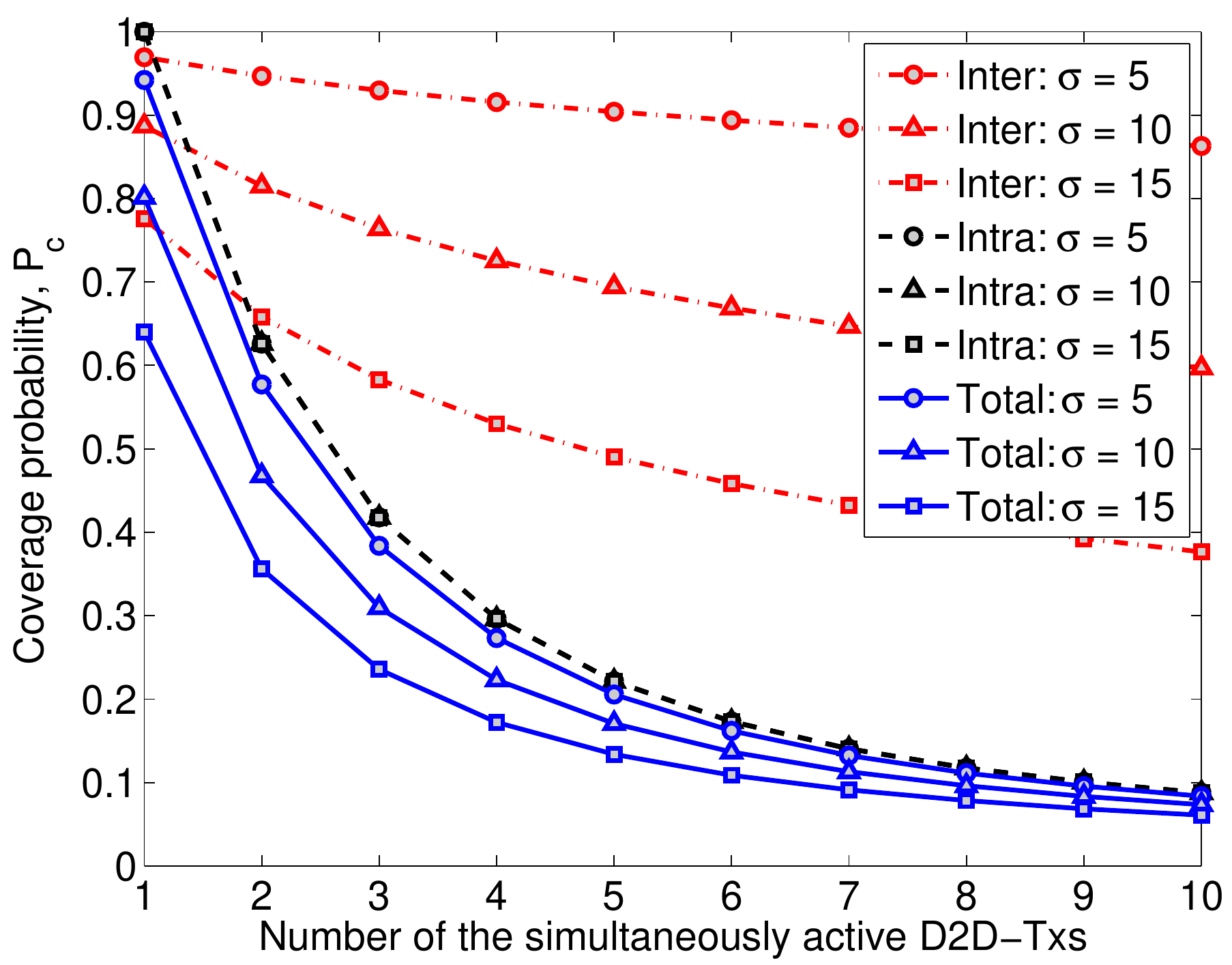}
              \caption{Coverage probability versus number of simultaneously active D2D transmitters for  $\lambda_\nrmc=150 \:\text{cluster}/\text{km}^2$ and different value of $\sigma$.}
                \label{Fig: Coverage probability versus number of simultaneously active D2D-Tx for  lambdac=150}
\end{minipage}%
\hfill
\noindent \begin{minipage}{.49\textwidth}
  \includegraphics[width=1\textwidth]{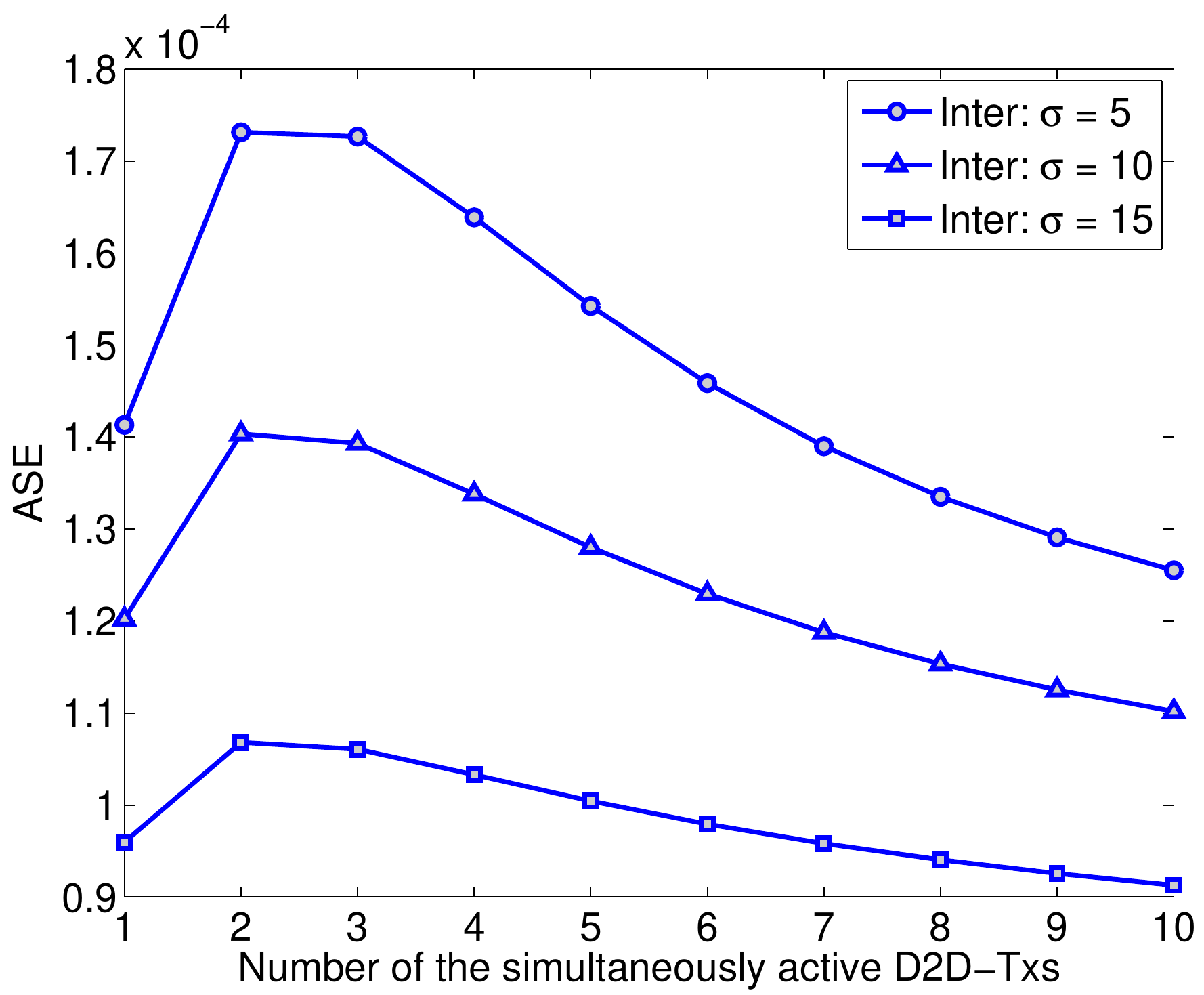}
              \caption{$\ase$  versus number of simultaneously active D2D transmitters for  $\lambda_\nrmc=150 \:\text{cluster}/\text{km}^2$ and different value of $\sigma$.}
       \label{Fig: ASE versus number of simultaneous active D2D-Tx for  lambdac=150}
\end{minipage}
\end{figure}
\begin{figure}
 \begin{minipage}{.49\textwidth}
  \includegraphics[width=1\textwidth]{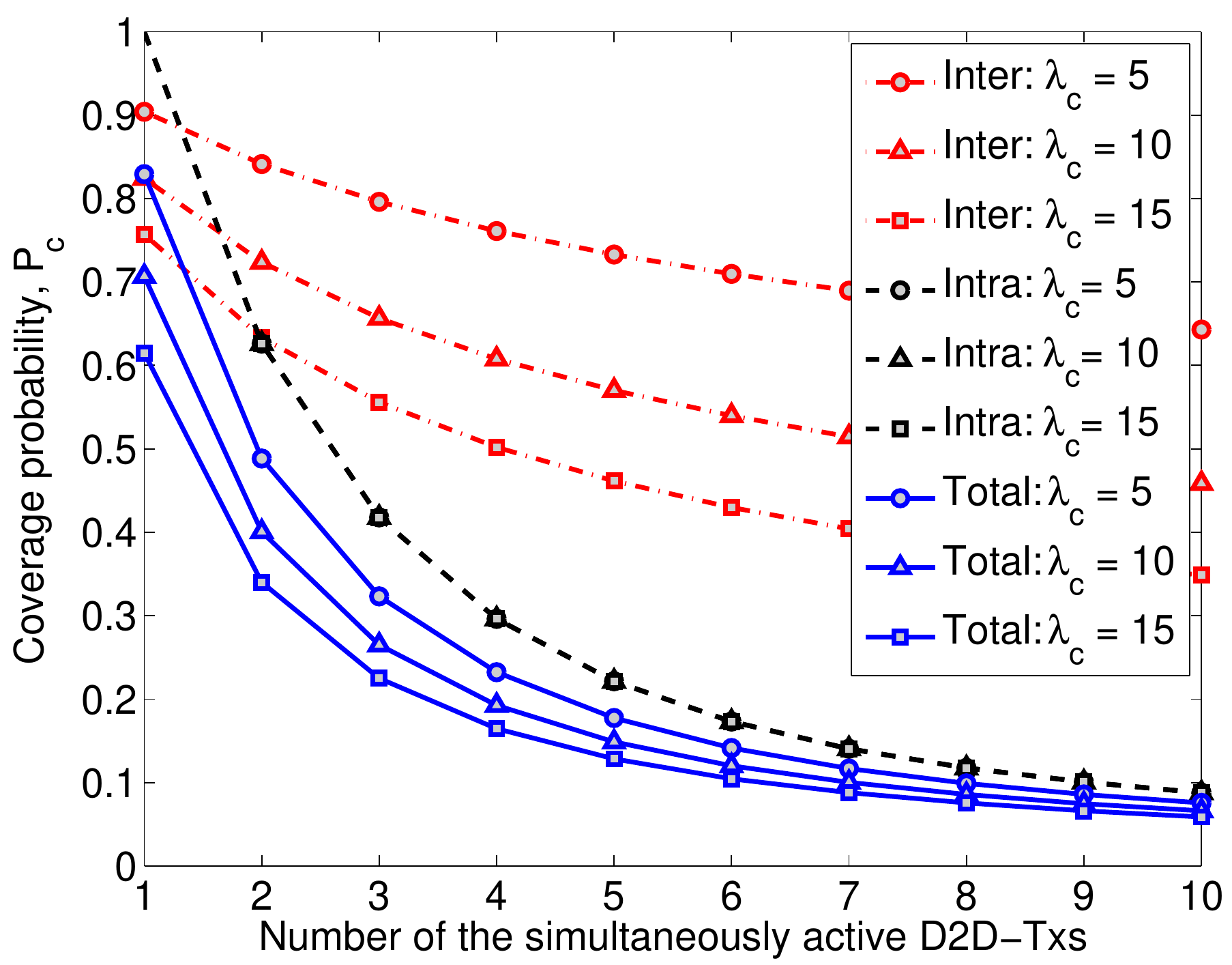}
              \caption{Coverage probability versus number of active D2D transmitters for  $\sigma=50$ and different value of $\lambda_\nrmc$.}
                \label{Fig: Coverage probability v.s. number of active D2D-Tx for different lambda sigma50}
\end{minipage}%
\hfill
 \begin{minipage}{.49\textwidth}
  \includegraphics[width=1\textwidth]{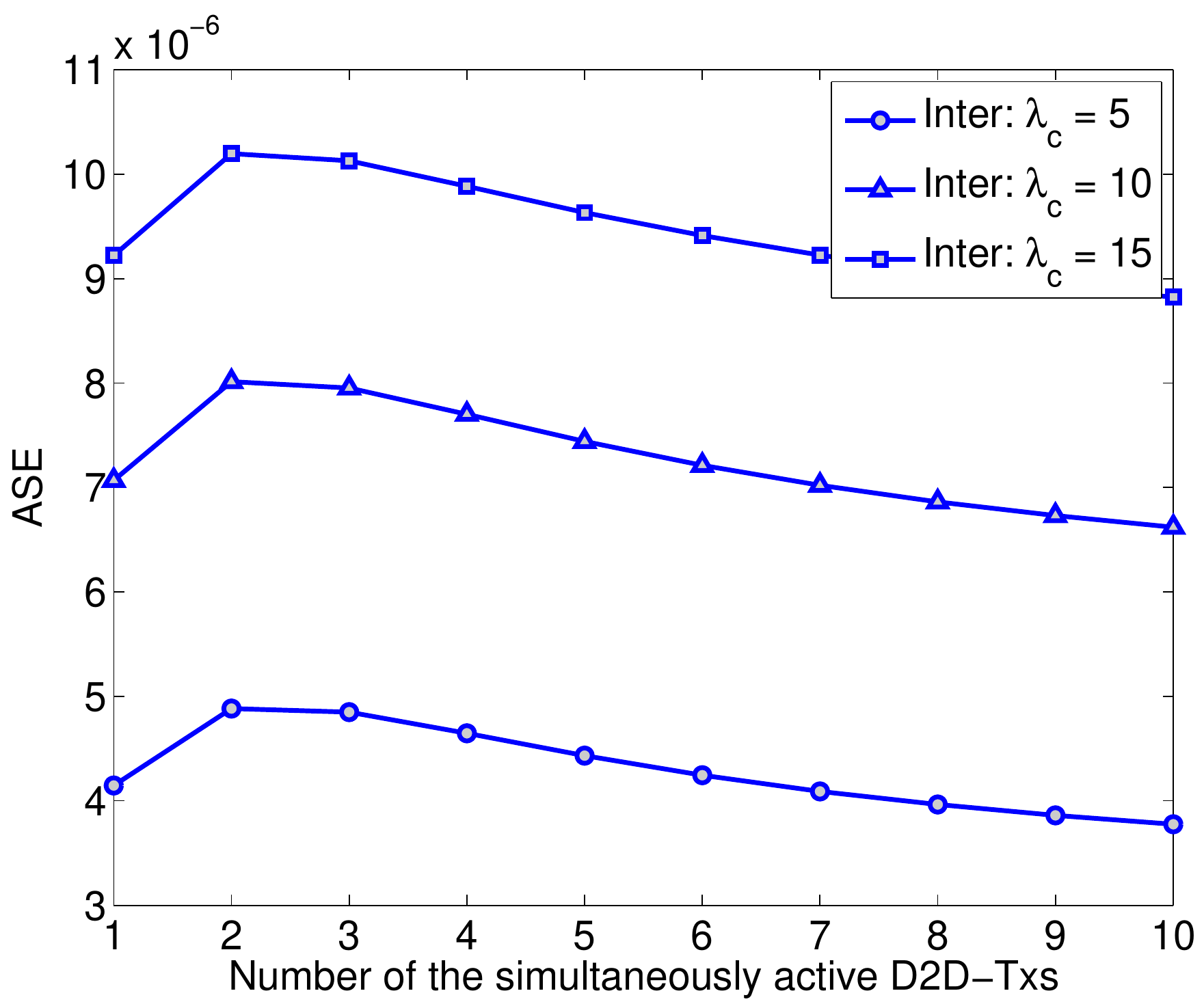}
              \caption{ $\ase$  versus number of active D2D transmitters for  $\sigma=50$ and different value of $\lambda_\nrmc$.}
                \label{Fig:  ASE v.s. number of active D2D-Tx for different lambda sigma50}
\end{minipage}
\end{figure}
\subsubsection{Optimum number of  simultaneously active links} 
\label{Sub: Optimum number of  simultaneous active link analysis}
In \figref{Fig: ASE versus number of simultaneous active D2D-Tx for  lambdac=150}, we present $\ase$  as a function of the number of simultaneously active D2D-Txs per cluster.
While more simultaneously active links may improve $\ase$, they may also increase interference significantly. Interestingly, it can be seen that the optimal value of $\ase$  is the same for a range of scattering variances.
This is because intra-cluster interference has a more dominant effect on this tradeoff  compared to the inter-cluster interference. 
Moreover,  as shown in \figref{Fig: Coverage probability versus number of simultaneously active D2D-Tx for  lambdac=150}, the coverage probability in the absence of the inter-cluster interference is independent of the scattering variance which leads to the same optimal value for a range of scattering variances. It can also be seen that smaller scattering variances result in higher $\ase$,  which highlights the importance of short range D2D communication. In 
\figref{Fig:  ASE v.s. number of active D2D-Tx for different lambda sigma50}, we present $\ase$ results for the clustered D2D network  with a large scattering variance for the different values of  cluster center densities.
By increasing the density of the  cluster centers,  both the density of the simultaneously active  transmitters inside the whole network, i.e., $\bar{m} \lambda_c$, and the inter-cluster interference are increased. 
While the former increases $\ase$, the latter decreases it. However, since inter-cluster interference doesn't have a dominant impact on the coverage probability, and hence on $\ase$, $\ase$ increases with the density of cluster centers. 
 A common observation in \figref{Fig: ASE versus number of simultaneous active D2D-Tx for  lambdac=150} and \figref{Fig:  ASE v.s. number of active D2D-Tx for different lambda sigma50} is again  the same optimal value for the average number of simultaneously active D2D transmitters per cluster.
As noted above, the intuition behind this behavior is that the intra-cluster interference, which has a significant impact on the tradeoff, is independent of both the scattering variance and the cluster center density. 

\begin{figure}
 \begin{minipage}{.49\textwidth}
  \includegraphics[width=1\textwidth]{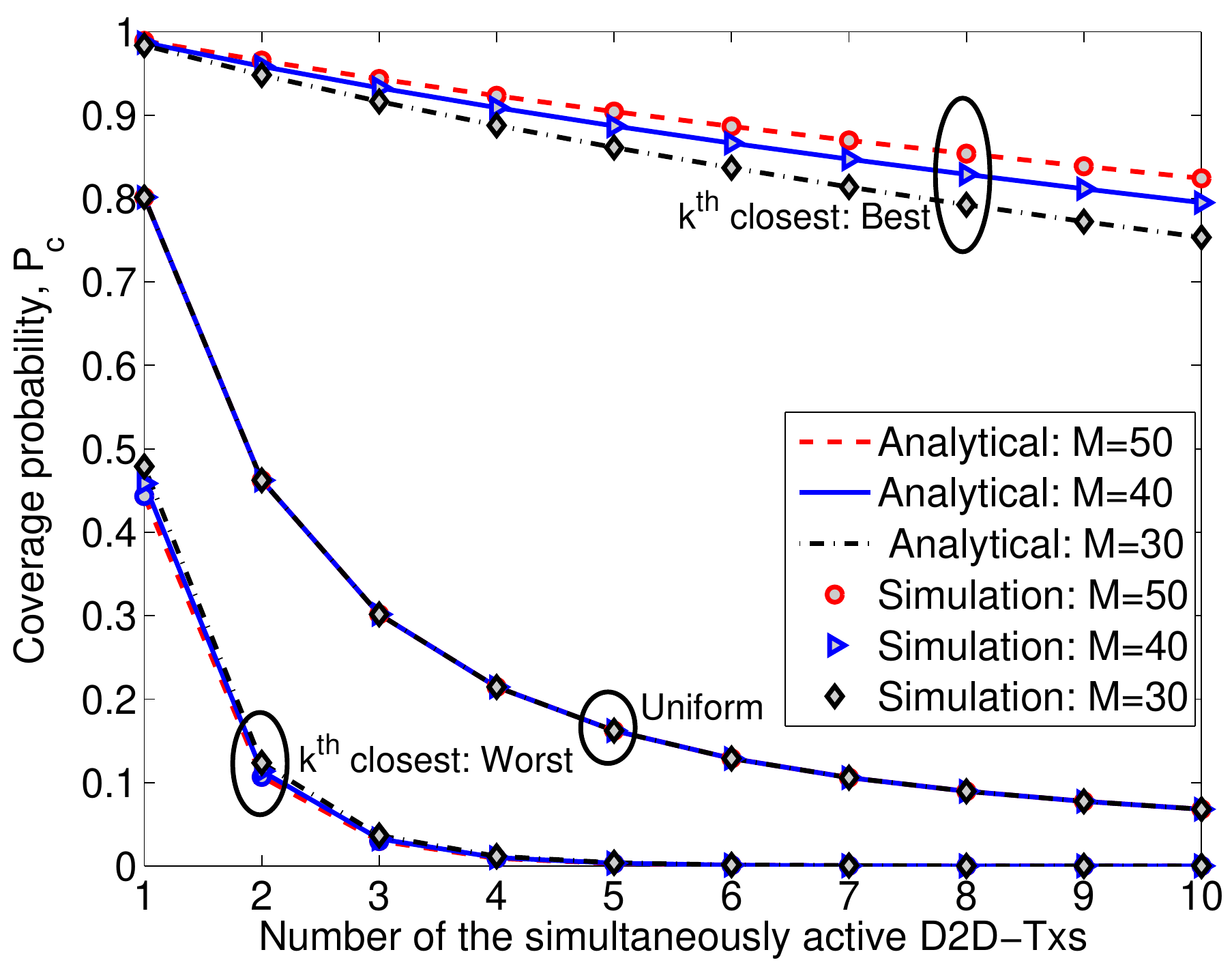}
              \caption{Coverage probability versus number of simultaneously active D2D transmitters for  $\lambda_\nrmc=150 \:\text{cluster}/\text{km}^2$ and various  values of transmitting cluster size $M$. }
                \label{Fig: Coverage versus number of simultaneous active D2D-Tx for  lambdac=150, max min random}
\end{minipage}%
\hfill
 \begin{minipage}{.525\textwidth}
  \includegraphics[width=1\textwidth]{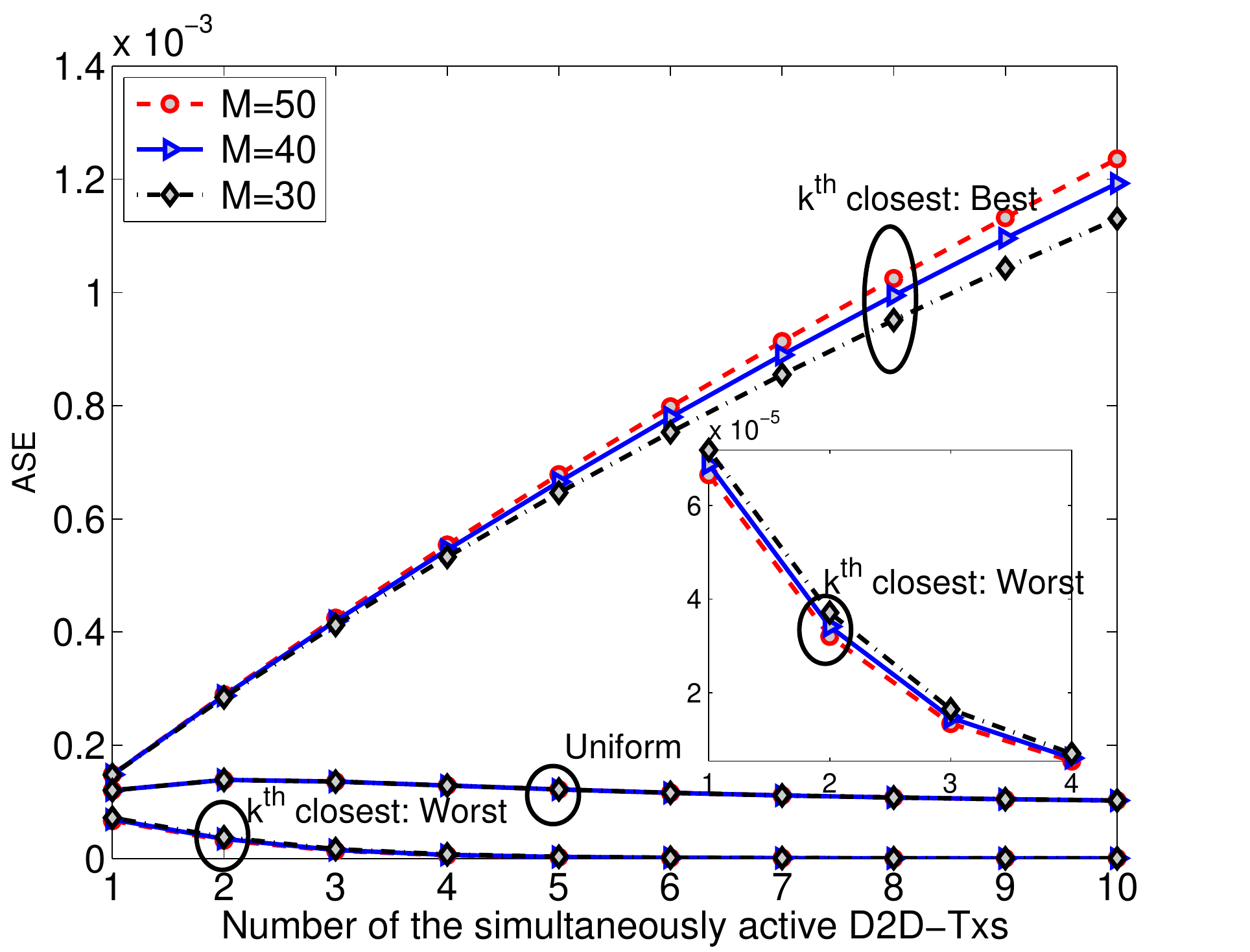}
              \caption{$\ase$  versus number of simultaneously active D2D transmitters for  $\lambda_\nrmc=150 \:\text{cluster}/\text{km}^2$and various  values of transmitting cluster size $M$.}
                \label{Fig: ASE versus number of simultaneous active D2D-Tx for  lambdac=150, max min random}
\end{minipage}
\end{figure}

\begin{figure}
\begin{minipage}{.49\textwidth}
  \includegraphics[width=1\textwidth]{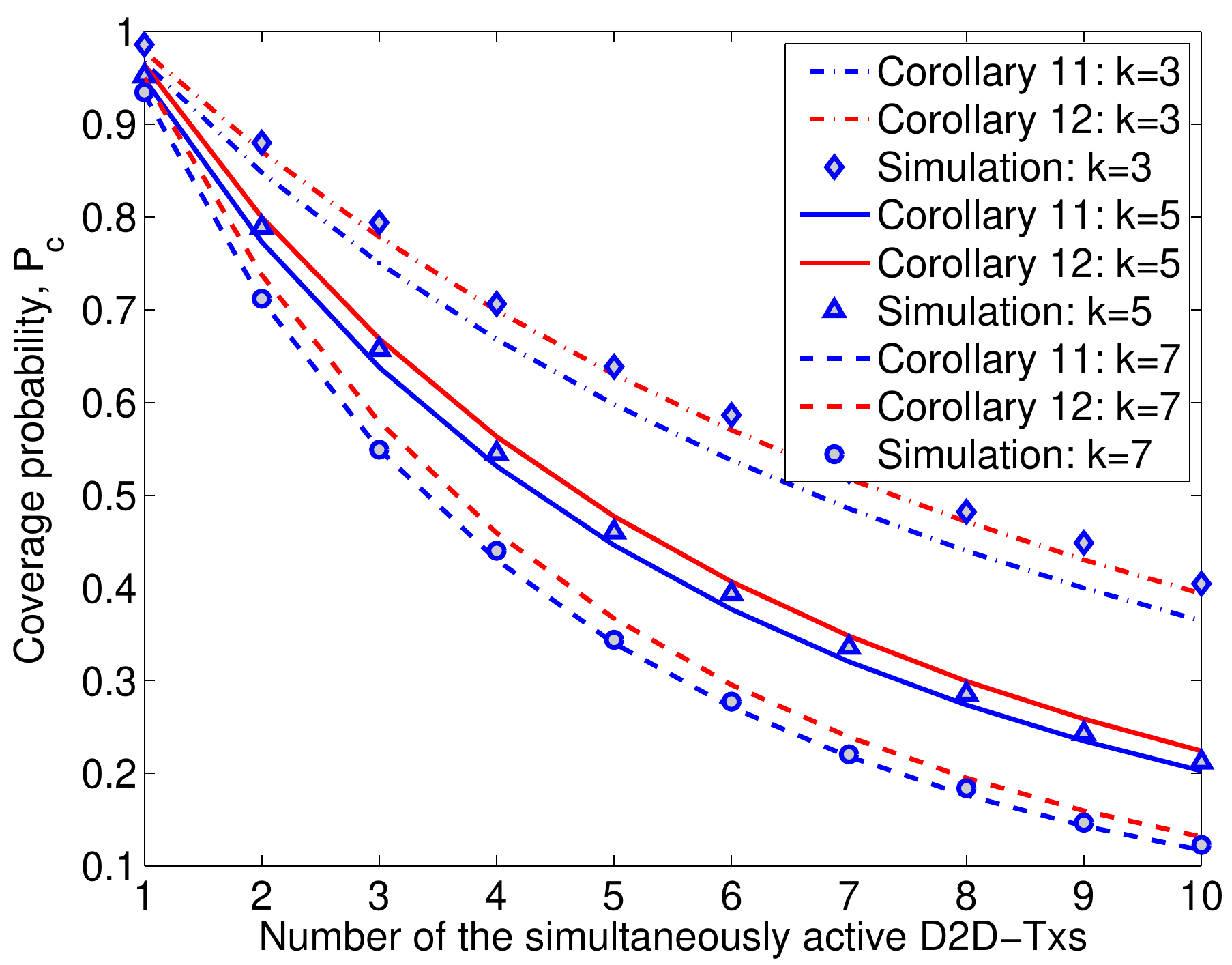}
              \caption{Coverage probability versus number of active D2D transmitters for  $\sigma=10$ and  $M=40$.}
                \label{Fig:  Coverage probability versus number of active D2D transmitters k-closest}
\end{minipage}%
\hfill
\begin{minipage}{.49\textwidth}
  \includegraphics[width=1\textwidth]{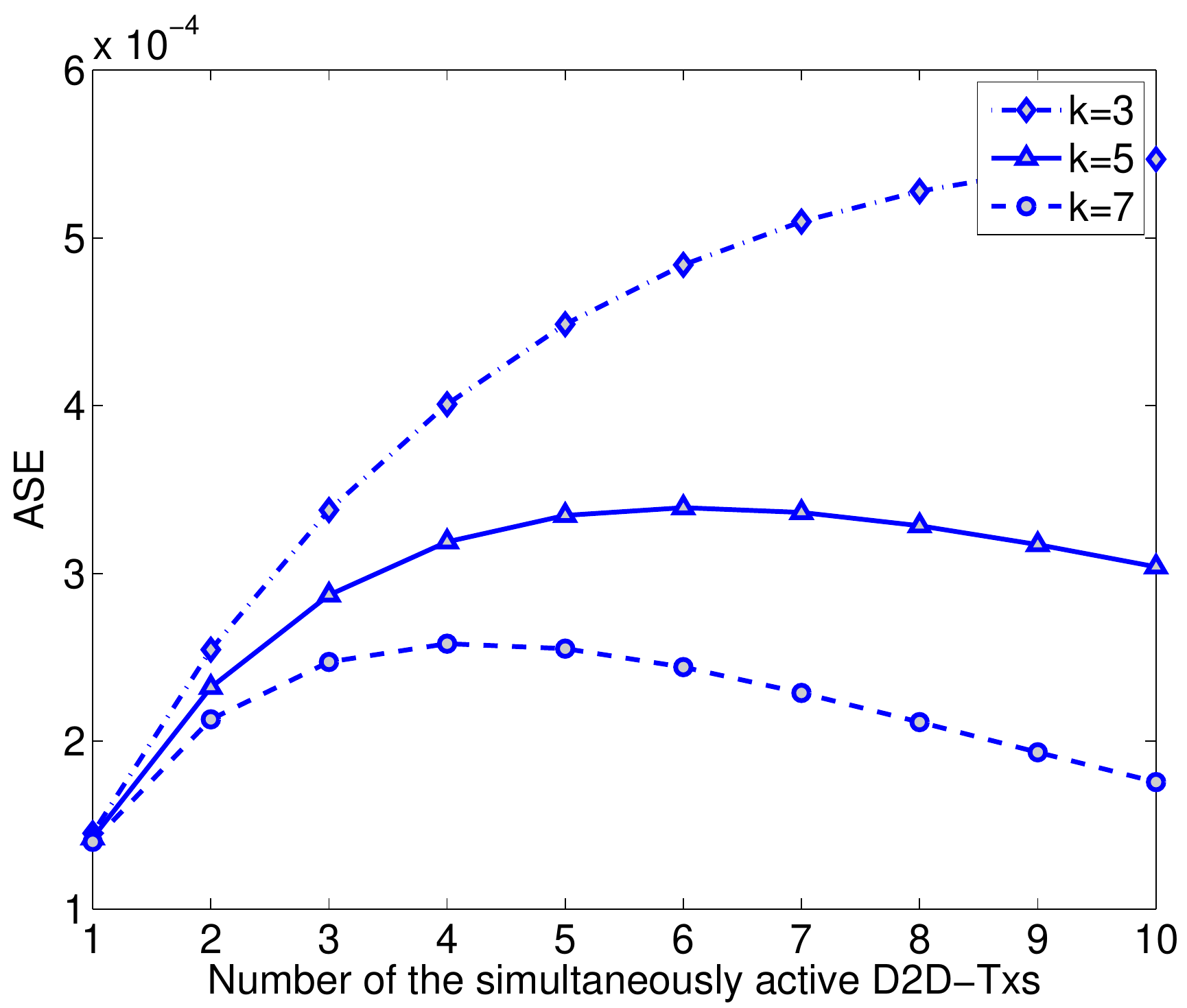}
              \caption{Analytical result for $\ase$ versus number of active D2D transmitters for  $\sigma=10$ and  $M=40$.}
                \label{Fig: ASE versus number of active D2D transmitters k-closest}
\end{minipage}
\end{figure}

\subsection{$k$-closest Content Availability}
\subsubsection{Best and worst link analyses}
After studying the performance of uniform content availability strategy in the previous subsection, we now compare it with the $k^{th}$ closest content availability strategy, in particular, its two special cases: (i) best link ($k=1$) and worst link ($k=M$) in Figs. \ref{Fig: Coverage versus number of simultaneous active D2D-Tx for  lambdac=150, max min random} and \ref{Fig: ASE versus number of simultaneous active D2D-Tx for  lambdac=150, max min random}. We first note that the analytical and simulation results match perfectly, thereby validating the analysis of the $k$-closest content availability case. Next, we notice that cluster size $M$ (number of possible transmitting devices per cluster) has a conflicting effect on the performance of best and worst link cases. This is simply because of ``order statistics'': larger cluster size $M$ decreases the minimum serving distance (best link) and increases the maximum serving distance (worst link). Another interesting observation can be made in \figref{Fig: ASE versus number of simultaneous active D2D-Tx for  lambdac=150, max min random}, where the $\ase$ in the best link case increases (in the range considered) with the number of simultaneously active links per cluster, thereby providing ``scalability'' to the D2D network.


\subsubsection{Tightness of the approximations}


Recall that the exact expression of intra-cluster interference for the $k$-closest content availability strategy derived in Lemma~\ref{lem: Laplace intra k-closest} involves two summations, which complicates the numerical evaluation of the exact coverage probability expression of Theorem~\ref{thm: coverage typical}. This motivated us to derive two approximations for the coverage probability in  Corollaries \ref{thm: coverage typical general} and \ref{cor: coverage typical approximation k closest} by approximating the Laplace transform of intra-cluster interference. These approximations are plotted along with the exact result in \figref{Fig:  Coverage probability versus number of active D2D transmitters k-closest}. We observe that both the Corollaries provide reasonably tight approximations, which get even better for higher values of $k$. The reason behind this trend is the fact that while deriving these approximations, we ignored the effect of excluding the $k^{th}$ closest device (serving) from the field of possible interferers. The effect of this approximation is more prominent at the small values of $k$.


\subsubsection{Optimum number of simultaneously active link} As shown in Figs. \ref{Fig:  Coverage probability versus number of active D2D transmitters k-closest} and \ref{Fig: ASE versus number of active D2D transmitters k-closest}, both coverage probability and $\ase$ increase significantly when the distance between the typical and serving devices is reduced, i.e., the value of $k$ is reduced. More interestingly, we note that the optimum number of simultaneously active D2D-Txs that maximizes $\ase$ also increases when the content of interest is made available closer to the typical device (smaller value of $k$). This highlights the importance of smarter content placement and scheduling in clustered D2D networks. 


\section{Conclusion}
\label{Conclusion}
In this paper, we developed the first comprehensive framework for D2D networks that is capable of capturing the fact that devices engaging in D2D may have multiple proximate devices any of which can act as a serving device. Modeling the device locations by a Poisson cluster process, we considered two content availability strategies that dictate the choice of the serving device: (i) {\em uniform content availability:} each device chooses its serving device uniformly at random from its own cluster, and (ii) {\em $k$-closest content availability:} each device connects to its $k^{th}$ closest device from its own cluster. Using tools from stochastic geometry, we derived exact expressions and several easy-to-use approximations for the coverage probability and area spectral efficiency. As a key intermediate step, we characterized the distributions of the distances from the typical device to various intra- and inter-cluster devices. To the best of our knowledge, this work is also the first to derive these distance distributions for the Thomas cluster process. 

Based on the aforementioned analyses, several interesting system design insights can be drawn. The most important amongst them is the existence of the optimal number of D2D transmitters that must be simultaneously activated in each cluster to maximize the area spectral efficiency. This can be interpreted as the classical tradeoff between higher interference power and more aggressive frequency reuse. Interestingly, we show that this optimal number of D2D transmitters is not too sensitive to the system parameters in the uniform content availability case. On the other hand, it increases significantly when the content of interest is biased to lie closer to the receivers in the $k$-closest content availability strategy. This highlights the importance of smart content placement and scheduling in these networks. Several other insights regarding the performance of this system as a function of cluster center density and cluster variance are also provided.

This work has several extensions both from the D2D network modeling and cluster point processes perspectives. For D2D network modeling, it is of interest to conduct measurement campaigns to gather spatial data for various types of user clusters at different locations, such as coffee shops, restaurants, airpots, and other hotspot zones, using which we can do data fitting to find classes of cluster processes that closely resemble real-world clusters. From the analytical perspective, an important extension would be to consider {\em in-band network} where devices use the same spectrum that is used by the cellular system. This brings forth several interesting coexistence issues, especially in the context of cluster point processes~\cite{AfsDhiC2015a}. Another interesting extension is to incorporate the effects of various content placement and content delivery strategies in the proposed model. Incorporating the effect of scheduling in these clustered networks presents another avenue of innovation. From the cluster point process perspective, it is important to derive the distribution of distances from the typical device to other devices for other point processes of interest, such as the  Mat\'ern cluster process. While the ``trends'' predicted in this work for D2D networks  are not likely to change with the choice of the cluster point process, these distance distributions will facilitate analysis of more general wireless networks with clustered nodes.
   

\appendix
\subsection{Proof of Lemma \ref{lem:Intra_cluster location _typical_general}}
\label{App: proof of Ricain}  
Denote the location of a device chosen uniformly at random in the representative cluster by $\mathbf{z}=\mathbf{x}_0+\mathbf{y} \in \nbbR^2$ and its distance from the typical device by $s = \| \mathbf{z} \| = \| \mathbf{x}_0+\mathbf{y}\| \in \nbbR_+$, where $\mathbf{x}_0=(x_1,x_2)$ and $\mathbf{z}=(z_1,z_2)$. Using the fact that $\mathbf{y}$ is a zero mean complex Gaussian random variable, it is easy to argue that $\rx=\ncalN(x_1,\sigma^2)$ and $\ry=\ncalN(x_2,\sigma^2)$ conditioned on $\mathbf{x}_0 = (x_1,x_2)$. The conditional joint distribution of $z_1$ and $z_2$ can thus be obtained as
\begin{align*}
f_\mathbf{Z}(z_1, z_2|\mathbf{x}_0) =  \frac{1}{{2 \pi} \sigma^2}\exp \left(-\frac{(\rx-x_1)^2}{2 \sigma^2} -\frac{(\ry-x_2)^2}{2 \sigma^2}\right).
\end{align*}
Since we are interested in the distribution of $S = \| \mathbf{Z} \|$, we define $z_1 = s \cos \theta$ and $z_2 = s \sin \theta$, where $\theta = \tan^{-1}\left(\frac{z_2}{z_1} \right)$. Using the usual transformation approach, we first evaluate
\begin{align*}
f_{S,\Theta} (s,\theta | \nbx_0) = f_\mathbf{Z}(z_1, z_2|\mathbf{x}_0) \times \left| \partial \left( \frac{z_1, z_2}{s,\theta} \right) \right|,
\end{align*}
where $\partial \left( \frac{z_1, z_2}{s,\theta} \right) = \left| \begin{array}{cc} 
\frac{\partial z_1}{\partial s} & \frac{\partial z_1}{\partial \theta}\\
\frac{\partial z_2}{\partial s} & \frac{\partial z_2}{\partial \theta}
\end{array} \right| = s$, using which $f_{S,\Theta} (s,\theta | \nbx_0)$ can be expressed as
\begin{align*}
f_{S,\Theta} (s,\theta | \nbx_0) &= \frac{s}{{2 \pi} \sigma^2}\exp \left(-\frac{(s \cos \theta-x_1)^2}{2 \sigma^2} -\frac{(s \sin \theta-x_2)^2}{2 \sigma^2}\right)\\
&\stackrel{(a)}{=}  \frac{s}{ \sigma^2} \exp \left( - \frac{s^2 + \nu_0^2}{2 \sigma^2} \right) \frac{1}{2 \pi} \exp \left( \frac{s x_1 \cos \theta + s x_2 \sin \theta}{\sigma^2} \right),
\end{align*}
where $\nu_0 = \sqrt{x_1^2+x_2^2}=\|\nbx_0\|$. Now the conditional marginal distribution of $S$ is $f_S(s| \nu_0) = $
\begin{align*}
\int_{0}^{2\pi} f_{S,\Theta} (s,\theta | \nbx_0) {\rm d} \theta = \frac{s}{ \sigma^2} \exp \left( - \frac{s^2 + \nu_0^2}{2 \sigma^2} \right) \underbrace{\int_{0}^{2\pi} \frac{1}{2 \pi} \exp \left( \frac{s x_1 \cos \theta + s x_2 \sin \theta}{\sigma^2} \right) {\rm d} \theta}_{I_0\left(\frac{s \nu_0}{\sigma^2}\right)},
\end{align*}
where conditioning on $\nu_0=\|\nbx_0\|$ instead of $\nbx_0$ suffices. This completes the derivation of $f_S(s|\nu_0)$. 

\subsection{Proof of Lemma \ref{OK_Last_Lemma}}
\label{App: proof of inter intra-cluster distance}
The joint distribution  of order statistics corresponding to the distances from intra-cluster devices which are closer than serving device to the typical  i.e., $\Bxx_{\rm in}$, conditioned on $s_{(k)}=r$ and $\nu_0=\|x_0\|$ is
$
f_{S_{(1)},...S_{(k-1)}|S_{(k)},\nu_0}(s_{(1)},...,s_{(k-1)}|s_{(k)},\nu_0)= \frac{f_{S_{(1)},...,S_{(k)}|\nu_0}(s_{(1)},...,s_{(k)}|\nu_0)}{f_{S_{(k)}}(s_{(k)}|\nu_0)},$
where $s_{(i)}\leq s_{(i+1)}
$. Here, the joint distribution is $f_{S_{(1)},...,S_{(k)}|\nu_0}(s_{(1)},...,s_{(k)}|\nu_0)$ 
\begin{align*}
&\stackrel{(a)}{=} \int..\int M!\prod_{i=1}^M f_{S}(s_{(i)}|\nu_0) \nrmd s_{(k+1)}... \nrmd s_{(M)}\\
&= M! \prod_{i=1}^k f_{S}(s_{(i)}|\nu_0)\int..\int \prod_{i=k+1}^M f_{S}(s_{(i)}|\nu_0) \nrmd s_{(k+1)}... \nrmd s_{(M)} \\
&\stackrel{(b)}{=} M! \prod_{i=1}^k f_{S}(s_{(i)}|\nu_0) \frac{1}{(M-k)!}\prod_{i=k+1}^M \int_{s_{(k)}}^\infty  f_{S}(s|\nu_0) \nrmd s 
\\&= \frac{M!}{(M-k)!} \prod_{i=1}^k f_{S}(s_{(i)}|\nu_0)(1-F_{S}(s_{(k)}|\nu_0))^{M-k},
\end{align*}
where (a) follows from the definition of the joint density function of $M$ order statistics \cite[eqn.~(2.10)]{ahsanullah2005order} with sampling distribution $f_S(s|\nu_0)$, and (b) follows from the symmetry of $\prod_{i=k+1}^M f_{S}(s_{(i)}|\nu_0)$ \cite[eqn.~(2.12)]{ahsanullah2005order}. Thus,
\begin{align}\notag
 f_{S_{(1)},...S_{(k-1)}|S_{(k)},\nu_0}(s_{(1)},...,s_{(k-1)}|s_{(k)},\nu_0)&= \frac{\frac{M!}{(M-k)!} \prod_{i=1}^k f_{S}(s_{(i)}|\nu_0)(1-F_{S}(s_{(k)}|\nu_0))^{M-k}}{f_{S_{(k)}}(s_{(k)}|\nu_0)}\label{eq: prod of inter distance}\\\notag
 &\stackrel{(c)}{=}(k-1)!\prod _{i=1}^{k-1}\frac{f_{S}(s_{(i)}|\nu_0)}{F_{S}(s_{(k)}|\nu_0)}, \quad s_{(i)}\leq s_{(i+1)}
\end{align}
where (c) follows by using the serving link distribution $f_{S_{(k)}}(s_{(k)}|\nu_0)$ given by \eqref{eq: k-closet dis to typical} with $r=s_{(k)}$.  
Note that there are $(k-1)!$ number of possible permutations of the ordered set $\{s_{(i)}\}_{{i=1}:{k-1}}$. Since the interfering devices are chosen uniformly at random, the permutation term $(k-1)!$  does not appear in the conditional joint distribution of the unordered set  $\{s_i\}_{i=1}^{k-1}$,
\begin{align}
 f_{S_{1},...S_{k-1}|S_{(k)},\nu_0}(s_{1},...,s_{k-1}|s_{(k)},\nu_0)=\prod _{i=1}^{k-1}\frac{f_{S}(s_{i}|\nu_0)}{F_{S}(s_{(k)}|\nu_0)}, \quad s_i<{s_{(k)}}
\end{align}
The product of distributions with the same functional form $\frac{f_{S}(s_{i}|\nu_0)}{F_{S}(s_{(k)}|\nu_0)}$ in the above expression  implies that the random variables of the unordered set $\{s_i, \forall s_i \le s_{(k)}\}$ are i.i.d.
Therefore, the sampling distribution of a set of $k-1$ i.i.d. random variables in the set $\{w_{\rm in}=\|x_0+y\| \} $ is
\begin{equation} \label{eq: f_w_in}
 f_{W_{\rm in}}(w_{\rm in}|\nu_0, r)=\left\{
 \begin{array}{cc}
 \frac{f_{S}(w_{\rm in}|\nu_0)}{F_{S}(r|\nu_0)}, & w_{\rm in}<s_{(k)}=r\\
 0, & w_{\rm in}\geq r
 \end{array}\right..
 \end{equation}
 The derivation of $f_{W_{\rm out}}(w_{\rm out}|\nu_0,r)$ follows  on the same lines 
and is hence skipped.
 
\subsection{Proof of Lemma \ref{lem: lap intra typical}}
\label{App: Proof of Lemma  lap intra typical}
The Laplace transform of  intra-cluster  interference at the typical device is
\begin{align}
\nonumber
\ncalL_\mathrm{I_{TX-cluster}}(s|\nu_0 )=&\E\Big[\exp\Big(-s\sum_{\jx\in \Bxx \setminus \tx} \hyxx\|\yjx+x_0\|^{-\alpha}\Big)\Big]\\\nonumber
=&\E_{\Bxx}\Big[\prod_{\jx\in \Bxx \setminus \tx} \E_{\hyxx}\left[\exp\left(-s \hyxx\|\yjx+x_0\|^{-\alpha}\right)\Big]\right]\\\nonumber
\stackrel{(a)}{=}&\E_{\Bxx}\Big[\prod_{\jx\in \Bxx \setminus \tx} \frac{1}{1+s \|\yjx+x_0\|^{-\alpha}}  \Big]\\\nonumber
 \stackrel{(b)} =& \sum_{k=0}^{M-1}\Big(\int_{\R^2} \frac{1}{1+s \|\yjx+x_0\|^{-\alpha}} f_Y(\yjx) \nrmd \yjx \Big)^k  \times \underbrace{\frac{(\bar{m}-1)^k e^{-(\bar{m}-1)}}{k! \xi}}_{\P(K=k|K<M-1)}
\end{align}
where (a) follows from the definition of  Laplace transform along with the fact that channel gain $\hyxx \sim \exp(1)$ (Rayleigh fading assumption), (b) follows from the fact that the locations of intra-cluster devices, conditioned on $x_0\in \Phi_\nrmc$, are independent and expectation over number of interfering devices which are Poisson distributed conditioned on total being less than $M-1$, where $\xi=\sum_{j=0}^{M-1}\frac{(\bar{m}-1)^je^{-(\bar{m}-1)}}{j!}$.
 Now under the assumption  $\bar{m} \ll M$, 
 the Laplace transform of intra-cluster interference reduces to
\begin{align}
&\exp\Big(-(\bar{m}-1)\int_{\R^2} \frac{s \|\yjx+x_0\|^{-\alpha}}{1+s \|\yjx+x_0\|^{-\alpha} } f_Y(\yjx) \nrmd y \Big)\ \label{eq: Lap intra}\\ \notag
&\stackrel{(a)}{=} \exp\Big(-(\bar{m}-1)\int_0^\infty \frac{s w^{-\alpha}}{1+s w^{-\alpha}}f_W(w|\nu_0)\nrmd w\Big),\notag
\end{align}
where (a)  follows from the change of variable $\|x_0+\yjx\| \rightarrow w$, and   converting coordinates from Cartesian to polar by using distance distribution given by Corollary \ref{lem:intra_cluster_uniform_interference}.
\subsection{Proof of Corollary \ref{Cor: Lap intra typical low bound}}
\label{App: proof of intra typical low bound}
Using conditional Laplace transform of intra-cluster interference in \eqref{eq: Lap intra}, the Laplace transform of intra-cluster interference can be independently de-conditioned as follows:
 \begin{align}\label{Eq: term2}\notag
\ncalL_\mathrm{I_{TX-cluster}}(s)=&\int_{\nbbR^2}\exp\Big(-(\bar{m}-1)\int_{\R^2} \frac{s \|\yjx+x_0\|^{-\alpha}}{1+s \|\yjx+x_0\|^{-\alpha} } f_Y(\yj)\nrmd \yj \Big)f_Y(x_0)\nrmd x_0\\\notag
  &\stackrel{(a)}{=}\int_{\nbbR^2}\exp\left(-(\bar{m}-1)\int_{\nbbR^2}\frac{s\|z\|^{-\alpha}}{1+s\|z\|^{-\alpha}} f_Y(z-x_0) \nrmd  z\right)f_Y(x_0)\nrmd x_0\\\notag
    &\stackrel{(b)}{\ge}\exp\left(-(\bar{m}-1)\int_{\nbbR^2}\int_{\nbbR^2}\frac{s\|z\|^{-\alpha}}{1+s\|z\|^{-\alpha}} f_Y(z-x_0)f_Y(x_0)\nrmd x_0 \nrmd  z\right)\\\notag
    &\stackrel{(c)}{\ge} \exp\left(-(\bar{m}-1)\int_{\nbbR^2}\frac{s\|z\|^{-\alpha}}{1+s\|z\|^{-\alpha}} \sup_{z \in \nbbR^2}(f_Y*f_Y)(z) \nrmd  z\right)\\\notag
    &\stackrel{(d)}{=}\exp\left(-\frac{\bar{m}-1}{4 \pi \sigma^2}\int_{\nbbR^2}\frac{s\|z\|^{-\alpha}}{1+s\|z\|^{-\alpha}}\nrmd  z\right)\\\notag
   & =\exp\left(- \pi \frac{\bar{m}-1}{4  \pi \sigma^2} s^{\frac{2}{\alpha}}  \frac{2 \pi / \alpha}{\sin (2 \pi / \alpha )} \right)
\end{align}
 where (a) simply follows from change of variable $x_0+\yjx \rightarrow z$, (b) follows from Jensen's inequality, (c) follows from the convolution definition, and (d) from Young's inequality \cite{folland2013real}.
\subsection{Proof of Lemma \ref{Lem: Lap_Inter}}
\label{APP:Proof of Lemma 2}
The Laplace transform of the aggregate interference  from the inter-cluster interferers at the typical device, $\ncalL_{I_\mathrm{\Psi_\nrmm\setminus Tx-cluster}} (s)     =\E\Big[\exp\Big(-s\sum_{x\in \Phi_{\nrmc}\setminus x_0}\sum_{\jx \in \Bx} \hyx\|x+\yj\|^{-\alpha}\Big)\Big]$  is equal to
\begin{align}\notag
&=\E_{\Phi_{\nrmc}} \Big[\prod_{x\in \Phi_{\nrmc}\setminus x_0}\E_{\Bx} \Big[\prod_{\jx \in \Bx}\E_{\hyx}\left[\exp(-s \hyx\|x+\yj\|^{-\alpha})\right]\Big]\Big] 
   \\\notag &\stackrel{(a)}{=}\E_{\Phi_{\nrmc}}\Big[\prod_{x\in \Phi_{\nrmc}\setminus x_0}\E_{\Bx} \Big[\prod_{\jx \in \Bx} \frac{1}{1+s\|x+\yj\|^{-\alpha}}\Big]\Big]\\\notag
 &\stackrel{(b)}=\E_{\Phi_{\nrmc}}\Big[\sum_{k=0}^{M} \Big(\int_{\R^2} \frac{1}{1+s \|\yjx+x\|^{-\alpha}} f_Y(\yjx) \nrmd \yjx \Big)^k 
 \P(K=k|K<M)\Big]\\\notag
&\stackrel{(c)} = \exp\Big(- \lambda_\nrmc\int_{\R^2} \Big(1-\sum_{k=0}^{M} \Big(\int_{\R^2} \frac{1}{1+s \|\yjx+x\|^{-\alpha}} f_Y(\yjx) \nrmd \yjx \Big)^k \frac{\bar{m}^k e^{-\bar{m}}}{k! \eta}\Big)\nrmd x\Big)\\ 
 &\stackrel{(d)} =\exp\Big(-2 \pi \lambda_\nrmc\int_0^\infty \Big(1-\sum_{k=0}^{M} \Big ( \int_0^\infty \frac{1}{1+s u^{-\alpha}}f_{U}(u|\nu)\nrmd u\Big)^k  \times \frac{\bar{m}^k e^{-\bar{m}}}{k! \eta}\Big)\nu \nrmd \nu \Big),
\end{align}     
where (a) follows from expectation over     $\hyx\sim\exp(1)$, (b) follows from the expectation over number of interfering devices per cluster, (c) follows from the probability generating functional (PGFL) of PPP \cite{stoyanstochastic}, and (d) follows by converting from Cartesian to polar coordinates. Now, under the assumption of $\bar{m}\ll M$, the Laplace transform of inter-cluster interference is:   
 \begin{align*}  
     &\ncalL_{I_\mathrm{\Psi_\nrmm\setminus Tx-cluster}} (s)  \stackrel{}{=}\exp\Big(-2 \pi \lambda_\nrmc\int_0^\infty \Big(1-\exp\Big(-\bar{m}\int_0^\infty \frac{s u^{-\alpha}}{1+s u^{-\alpha}} f_U(u|\nu)\nrmd u \Big)\Big)\nu \nrmd \nu\Big).
\end{align*}

\subsection{Proof of Corollary \ref{cor: worst case inter lower}}
\label{App: Proof of Cor1}
The lower bound can be derived as follows:
\begin{align}
    \ncalL_{I_\mathrm{\Psi_\nrmm\setminus Tx-cluster}}(s)
    &\stackrel{(a)}{ \ge}
    \exp\Big(-2 \pi \lambda_\nrmc\int_0^\infty \Big(\bar{m} \int_0^\infty \frac{s u^{-\alpha}}{1+s u^{-\alpha}} f_U(u|\nu)\nrmd u \nu \nrmd \nu\Big)\Big)\notag\\
   &\stackrel{(b)}{=}\exp\Big(-2 \pi \lambda_\nrmc \Big(\bar{m} \int_0^\infty \frac{s u^{-\alpha}}{1+s u^{-\alpha}}u\nrmd u \Big)\Big)\\\notag
          &=  \exp\left(- \pi \lambda_\nrmc \bar{m} s ^{2/\alpha} \frac{2 \pi / \alpha}{\sin (2 \pi / \alpha )} \right)\notag
\end{align}
where (a) follows from the exponential Taylor series expansion and the fact that $1-\exp(-a x)\le a, \: a\ge 0$, and (b)  follows from the Rician distribution property that $\int_0^\infty f_U(u|\nu) \nu \nrmd \nu=u$, where $f_U(u|\nu)$ is given by  Lemma \ref{Lem: Inter cluster distance}.
\subsection{Proof of Lemma \ref{lem: Laplace intra k-closest}}
\label{App: Laplace intra k-closest}
The Laplace transform of interference from disjoint sets of intra-cluster interferers $\ncalB^{x_0}_{\rm in}$ and $\ncalB^{x_0}_{\rm out}$ is $\ncalL_{ I_{\mathrm{Tx-cluster}}}(s,r|\nu_0)$
\begin{align*}
& =\E \left[ \exp\left(- s\left( \sum_{\jx\in {\mathcal{B}}^{x_0}_{\rm in} }  \hyxx\|x_0+\yjx\|^{-\alpha}+\sum_{\jx\in {\mathcal{B}}^{x_0}_{\rm out} }  \hyxx\|x_0+\yjx\|^{-\alpha}
\right)\right) \right]\\
&\stackrel{(a)}{=}\E \bigg[ \prod_{\jx\in {\mathcal{B}}^{x_0}_{\rm in} } \exp\left(- s \hyxx\|x_0+\yjx\|^{-\alpha} \right) \prod_{\jx\in {\mathcal{B}}^{x_0}_{\rm out} }  \exp \left( - s \hyxx\|x_0+\yjx\|^{-\alpha}
\right) \bigg]\\
&\stackrel{(b)}{=}\E \bigg[ \prod_{\jx\in {\mathcal{B}}^{x_0}_{\rm in} } \frac{1}{1+s \|x_0+\yjx\|^{-\alpha} } \prod_{\jx\in {\mathcal{B}}^{x_0}_{\rm out} }  \frac{1}{1+s \|x_0+\yjx\|^{-\alpha}
} \bigg]\\
&\stackrel{(c)}{=}\E\bigg[\sum_{l=0}^{g_{\rm m}} \underbrace{{n  \choose l}  \left(\frac{k-1}{M-1}\right)^{l}  \left(\frac{M-k}{M-1}\right)^{n-l} \frac{1}{I_{1-\frac{k-1}{M-1}}(n-g_{\rm m},g_{\rm m}+1)}}_{\P(l|L\le g_{\rm m})} 
\\ & \bigg(\underbrace{\int_0^r \frac{1}{1+s w_{\rm in}^{-\alpha} }f_{W_{\rm in}}(w_{\rm in}|\nu_0, r) \nrmd w_{\rm in}}_{\ncalK_{\rm in}(s, r|\nu_0)} \bigg)^l  \bigg(\underbrace{\int_r^{\infty} \frac{1}{1+s w_{\rm out}^{-\alpha} }f_{W_{\rm out}}(w_{\rm out}|\nu_0,r) \nrmd w_{\rm out}}_{\ncalK_{\rm out}(s, r|\nu_0)} \bigg)^{n-l}\bigg]\\
&\stackrel{(d)}{=}\sum_{n=0}^{M-1} \sum_{l=0}^{g_{\rm m}} \P(l|L\le g_{\rm m})   \ncalK_{\rm in}(s, r|\nu_0)^l  \ncalK_{\rm out}(s, r|\nu_0) ^{n-l}  \frac{(\bar{m}-1)^n e^{-(\bar{m}-1)}}{n! \xi},
\end{align*}
with $g_{\rm m}=\min(n,k-1)$, and $\xi=\sum_{j=0}^{M-1}\frac{(\bar{m}-1)^je^{-(\bar{m}-1)}}{j!}$, where  (a) follows from the fact that $\ncalB^{x_0}_{\rm in}$ and $\ncalB^{x_0}_{\rm out}$ are two disjoint sets, (b) follows from the definition of Laplace transform and the fact that $\hyxx\sim \exp(1)$, (c) follows from the expectation over the number of devices in $\ncalB^{x_0}_{\rm in}$ and $\ncalB^{x_0}_{\rm out}$, where the number of devices in $\ncalB^{x_0}_{\rm in}$ is truncated binomial distribution due to the fact that $l\le\min(n,k-1)$, along with the fact that distances from interfering devices to the typical device conditioned on $r$ and $\nu_0$ in each set are i.i.d., and (d) from the fact that number interfering devices is Poisson distributed with mean $\bar{m}-1$ conditioned on the total being less than $M-1$.%

\section*{Acknowledgment}
The authors would like to thank  Zeinab Yazdanshenasan and  SaiDhiraj Amuru for their helpful feedback.

\bibliographystyle{IEEEtran}
\bibliography{ref,Dhillon_MA-J1}
\end{document}